\documentclass{article}
\usepackage[left=1in, right=1in, top=1in, bottom=1in]{geometry}
\usepackage{enumitem}
\usepackage{authblk}
\usepackage{graphicx}	
\usepackage{amsmath, amsthm, amssymb}
\usepackage{xspace}
\usepackage{comment}
\usepackage{hyperref}
\usepackage[dvipsnames]{xcolor}
\usepackage{algorithm}
\usepackage{algpseudocode}
\usepackage{todonotes}
\usepackage{mathtools}
\usepackage[capitalize]{cleveref}
\usepackage{ifthen}
\usepackage[mathlines]{lineno}

\newcommand{\cc}[1]{\ensuremath{\mathsf{#1}}}
\newcommand{\algprobm}[1]{\textsc{#1}\xspace}

\newcommand{\betacc}[1]{\ifthenelse{\equal{#1}{1}}{\exists^{\log n}}{\exists^{\log^{#1}n}}} 
\newcommand{\alphacc}[1]{\ifthenelse{\equal{#1}{1}}{\forall^{\log n}}{\forall^{\log^{#1}n}}} 

\newcommand{\Soc}{\text{Soc}}
\newcommand{\Fac}{\text{Fac}}
\newcommand{\F}{\mathbb{F}}

\renewcommand{\setminus}{\mysetminus}
\newcommand{\mysetminusD}{\raisebox{.8pt}{\hbox{\tikz{\draw[line width=0.6pt,line cap=round] (3.5pt,0pt) -- (0,5.2pt);}}}}
\newcommand{\mysetminusT}{\mysetminusD}
\newcommand{\mysetminusS}{\raisebox{.5pt}{\hbox{\tikz{\draw[line width=0.45pt,line cap=round] (2.2pt,0) -- (0,3.8pt);}}}}
\newcommand{\mysetminusSS}{\raisebox{.35pt}{\hbox{\tikz{\draw[line width=0.4pt,line cap=round] (1.5pt,0) -- (0,2.8pt);}}}}

\newcommand{\mysetminus}{\mathbin{\mathchoice{\mysetminusD}{\mysetminusT}{\mysetminusS}{\mysetminusSS}}}

\theoremstyle{plain}
\newtheorem{theorem}{Theorem}[section]
\newtheorem{proposition}[theorem]{Proposition}
\newtheorem{corollary}[theorem]{Corollary}
\newtheorem{lemma}[theorem]{Lemma}

\newtheorem{observation}[theorem]{Observation}

\theoremstyle{definition}
\newtheorem{definition}[theorem]{Definition}
\newtheorem{remark}[theorem]{Remark}
\newtheorem{question}[theorem]{Problem}

\newcommand{\wt}{\text{wt}}

\DeclareMathOperator{\Inn}{Inn}
\DeclareMathOperator{\Aut}{Aut}
\DeclareMathOperator{\ncl}{ncl}

\DeclareMathOperator{\rad}{Rad}

\DeclareMathOperator{\poly}{poly}

\DeclareMathOperator{\rk}{rk}

\newcommand*{\ComplexityClass}[1]{\ensuremath{\mathsf{#1}}\xspace}

\newcommand*{\LogSpace}{\ComplexityClass{L}}

\newcommand{\AC}{\ComplexityClass{AC}}
\newcommand{\SAC}{\ComplexityClass{SAC}}

\newcommand{\FOLL}{\ComplexityClass{FOLL}}
\newcommand{\FOPLL}{\ensuremath{\textsf{FOLL}^{O(1)}}}

\title{On the Complexity of Identifying Groups without Abelian Normal Subgroups: Parallel, First Order, and GI-Hardness}

\author[1,2]{Joshua A. Grochow}
\author[3]{Dan Johnson}
\author[3]{Michael Levet}
\affil[1]{Department of Computer Science, University of Colorado Boulder}
\affil[2]{Department of Mathematics, University of Colorado Boulder}
\affil[3]{Department of Computer Science, College of Charleston}

\begin{document}
\maketitle
\begin{abstract}
In this paper, we exhibit an $\textsf{AC}^{3}$ isomorphism test for groups without Abelian normal subgroups (a.k.a. Fitting-free groups), a class for which isomorphism testing was previously known to be in $\mathsf{P}$ (Babai, Codenotti, and Qiao; ICALP '12). Here, we leverage the fact that $G/\text{PKer}(G)$ can be viewed as permutation group of degree $O(\log |G|)$. As $G$ is given by its multiplication table, we are able to implement the solution for the corresponding instance of \algprobm{Twisted Code Equivalence} in $\textsf{AC}^{3}$.

In sharp contrast, we show that when our groups are specified by a generating set of permutations, isomorphism testing of Fitting-free groups is at least as hard as \algprobm{Graph Isomorphism} and \algprobm{Linear Code Equivalence} (the latter being \algprobm{GI}-hard and having no known subexponential-time algorithm).

Lastly, we show that any Fitting-free group of order $n$ is identified by $\textsf{FO}$ formulas (without counting) using only $O(\log \log n)$ variables. This is in contrast to the fact that there are infinite families of Abelian groups that are not identified by \textsf{FO} formulas with $o(\log n)$ variables (Grochow \& Levet, FCT '23).
\end{abstract}

\thispagestyle{empty}

\newpage

\setcounter{page}{1}

\section{Introduction}

The \algprobm{Group Isomorphism} problem (\algprobm{GpI}) takes as input two finite groups $G$ and $H$, and asks if there exists an isomorphism $\varphi : G \to H$. When the groups are given by their multiplication (Cayley) tables, it is known that $\algprobm{GpI}$ belongs to $\textsf{NP} \cap \textsf{coAM}$. The generator-enumeration strategy has time complexity $n^{\log_{p}(n) + O(1)}$, where $n$ is the order of the group and $p$ is the smallest prime dividing $n$. This strategy was independently discovered by Felsch and Neub\"user \cite{FN} and Tarjan (see \cite{MillerTarjan}). The parallel complexity of the generator-enumeration strategy has been gradually improved \cite{LiptonSnyderZalcstein, Wolf, ChattopadhyayToranWagner, TangThesis}. Collins, Grochow, Levet, and Weiß recently improved the bound for the generator enumeration strategy to
\[
\exists^{\log^{2} n} \forall^{\log n} \exists^{\log n}\textsf{DTISP}(\text{polylog}(n), \log(n)),
\]
which can be simulated by depth-$4$ $\textsf{quasiAC}^{0}$ circuits of size $n^{O(\log^4 n)}$ \cite{CGLWISSAC}. Algorithmically, the best known bound for \algprobm{GpI} is $n^{(1/4)\log(n)+O(1)}$, due to Rosenbaum \cite{Rosenbaum2013BidirectionalCD} and Luks \cite{LuksCompositionSeriesIso} (see \cite[Sec. 2.2]{GR16}). Even the impressive body of work on practical algorithms for this problem in more succinct input models, led by Eick, Holt, Leedham-Green and O'Brien (e.\,g., \cite{BEO02, ELGO02, BE99, CH03}) still results in an $n^{\Theta(\log n)}$-time algorithm in the general case (see \cite[Page 2]{WilsonSubgroupProfiles}).

In practice, such as working with computer algebra systems, the Cayley model is highly unrealistic. Instead, the groups are given by generators as permutations or matrices, or as black-boxes. In the setting of permutation groups, \algprobm{GpI} belongs to $\textsf{NP}$ \cite{LuksReduction}. When the groups are given by generating sets of matrices, or as black-boxes, $\algprobm{GpI}$ belongs to $\textsf{Promise}\Sigma_{2}^{p}$ \cite{BabaiSzemeredi}; it remains open as to whether $\textsf{GpI}$ belongs to $\textsf{NP}$ or $\textsf{coNP}$ in such succinct models.  In the past several years, there have been significant advances on algorithms with worst-case guarantees on the serial runtime for special cases of this problem; we refer to \cite{GQCoho, DietrichWilson, GrochowLevetWL} for a survey.

Key motivation for \algprobm{GpI} comes from its relationship to the \algprobm{Graph Isomorphism} problem (\algprobm{GI}). When the groups are given verbosely by their multiplication tables, \algprobm{GpI} is $\textsf{AC}^{0}$-reducible to $\algprobm{GI}$ \cite{ZKT}. On the other hand, there is no reduction from \algprobm{GI} to \algprobm{GpI} computable by $\textsf{AC}$ circuits of depth $o(\log n / \log \log n)$ and size $n^{\text{polylog}(n)}$ \cite{ChattopadhyayToranWagner}. In light of Babai's breakthrough result that $\algprobm{GI}$ is quasipolynomial-time solvable \cite{BabaiGraphIso}, $\algprobm{GpI}$ in the Cayley model is a key barrier to improving the complexity of $\algprobm{GI}$. There is considerable evidence suggesting that $\algprobm{GI}$ is not $\textsf{NP}$-complete \cite{Schoning, BuhrmanHomer, ETH, BabaiGraphIso, GILowPP, ArvindKurur}. As verbose $\algprobm{GpI}$ reduces to $\algprobm{GI}$, this evidence also suggests that $\algprobm{GpI}$ is not $\textsf{NP}$-complete. 

There is a natural divide-and-conquer strategy for \algprobm{GpI}, which proceeds as follows. Given groups $G_1, G_2$, we first seek to decompose the groups in terms of normal subgroups $N_{i} \trianglelefteq G_i$ ($i \in [2]$) and the corresponding quotients $G_{i}/N_{i}$. It is then necessary (but not sufficient) to decide if $N_{1} \cong N_{2}$ and $G_{1}/N_{1} \cong G_{2}/N_{2}$. The solvable radical of a finite group $G$ is the (unique) largest normal solvable subgroup of $G$; recall that solvability can be defined inductively by saying that a group is solvable if it is either Abelian, or if it has a normal Abelian subgroup $N$ such that $G/N$ is solvable. Thus every finite group can be decomposed into its solvable radical $N= \rad(G)$ and the corresponding quotient $G/\rad(G)$. The latter quotient has no Abelian normal subgroups; we will refer to the class of such groups as \emph{Fitting-free}.\footnote{The terms \textit{Fitting-free} or \textit{trivial-Fitting} are common in the group theory literature (see, e.\,g., \cite{CH03, Holt2005HandbookOC}) while the term \textit{semisimple} has been used in the complexity theory literature (see, e.\,g., \cite{BCGQ, BCQ}).} Fitting-free groups are thus a natural class of groups in a variety of settings, and serve not only as a necessary special case for group isomorphism testing, but also as a potentially useful stepping stone towards the general case (e.\,g., \cite{CH03, GQCoho, GR16}).

The class of Fitting-free groups has received significant attention. In the multiplication (Cayley) table model, it took a series of two papers \cite{BCGQ, BCQ} to establish a polynomial-time isomorphism test for this family. In the setting of permutation groups, Cannon and Holt \cite{CH03} exhibited an efficient practical isomorphism test. Fitting-free groups have also been investigated in the context of the Weisfeiler--Leman (WL) algorithm. Grochow and Levet \cite{GLDescriptiveComplexity} showed that this family is identified by the $7$-dimensional $2$-ary WL algorithm using only $7$ rounds. Brachter \cite{BrachterThesis} subsequently showed that this family has $1$-ary WL-dimension $O(\log \log n)$. Recently, Das and Thakkar \cite{DasThakkarMPD} exhibited a Las Vegas polynomial-time algorithm to compute a minimum faithful permutation representation for Fitting-free groups, when they are given succinctly by generating sequences of permutations. 

\noindent \\ \textbf{Main Results.} In this paper, we investigate the complexity of identifying Fitting-free groups. Our first main result is the following.

\begin{theorem} \label{thm:main}
Isomorphism of Fitting-free groups given by their multiplication tables can be decided in $\mathsf{AC}^3$. More specifically, there is a uniform $\mathsf{AC}^3$ algorithm that, given two groups $G,H$ of order $n$, decides if both are Fitting-free, and if they are, correctly decides whether they are isomorphic.
\end{theorem}

There are two key motivations for Theorem~\ref{thm:main}. The first comes from \textsc{GI}, whose best known lower-bound is $\textsf{DET}$ \cite{Toran}, which contains $\textsf{NL}$ and is a subclass of $\textsf{TC}^{1}$. There is a large body of work spanning almost 40 years on $\textsf{NC}$ algorithms for \algprobm{GI}, see  \cite{LevetRombachSieger} for a survey. In contrast, the work on $\textsf{NC}$ algorithms for \textsc{GpI} is very new, and $\textsc{GpI}$ is strictly easier than \textsc{GI} under $\textsf{AC}^{0}$-reductions \cite{ChattopadhyayToranWagner}. So it is surprising that we know more about parallelizing subclasses of \textsc{GI} than subclasses of \textsc{GpI}.

The second comes from the fact that complexity-theoretic lower bounds for \algprobm{GpI} remain wide open, even against depth-$2$ $\textsf{AC}$ circuits. In contrast, \algprobm{GpI} is solvable using depth-$4$ $\textsf{quasiAC}^{0}$ circuits \cite{CGLWISSAC}. Isomorphism testing is in $\textsf{P}$ for a number of families of groups, and closing the gap between $\textsf{P}$ and $\textsf{AC}^0$ for special cases is a key step in trying to characterize the complexity of the general \algprobm{GpI} problem. Fitting-free groups are not merely a special case, but an important structural stepping stone on the way to the general problem (see e.g., \cite{CH03,GR16}).

In sharp contrast, in the model of permutation groups we establish the following lower bound.

\begin{theorem} \label{thm:GpIPermutationGroups}
\algprobm{Graph Isomorphism}---and even \algprobm{Linear Code Equivalence}---is polynomial-time reducible to isomorphism testing of Fitting-free groups, when the groups are specified by generating sequences of permutations.
\end{theorem}

While the hardness afforded by \textsf{GI}-hardness is tempered by Babai's result \cite{BabaiGraphIso} that \algprobm{GI} is in \textsf{QuasiP}, for \algprobm{Linear Code Equivalence}, no algorithm with running time $2^{o(n)}$ is currently known (that it is solvable in time $2^{O(n)}$ was established by Babai in \cite{BCGQ}; recently extended to \textsc{Monomial CodeEq}\footnote{We note that the problems called \textsc{Linear CodeEq} and \textsc{Monomial CodeEq} in the complexity of \textsc{Group Isomorphism} literature are respectively referred to as \textsc{Permutational CodeEq} and \textsc{Linear CodeEq} in the coding literature. In this paper we stick with the former convention.} and with improved randomized and quantum algorithms by Bennett, Bhatia, Biasse, Durisheti, LaBuff, Lavorante, \& Waitkevich \cite{BBBDLLW}, though even those improvements are only in the base of the exponent). \textsc{GI} is known to reduce to \textsc{Linear Code Equivalence} \cite{PR97, LuksReduction}.\footnote{Luks showed a related result using a construction which Miyazaki \cite{Miyazaki} then showed also gave this result about \textsc{Linear Code Equivalence}.}

\begin{remark}
Luks \cite{LuksReduction} indicates that \algprobm{Graph Isomorphism} is polynomial-time reducible to isomorphism testing of permutation groups. In particular, Luks exhibits a reduction from computing the normalizer $N_{G}(H)$ to the \algprobm{Conj-Group} problem, the latter of which takes as input $G, H_1, H_2 \leq \text{Sym}(\Omega)$ and asks if there exists $g \in G$ such that $H_{1}^{g} = H_{2}$. Luks then cited a paper in-preparation (citation [3] in \cite{LuksReduction}), for a reduction from \algprobm{Conj-Group} to isomorphism testing of permutation groups. To the best of our knowledge, this reduction from \algprobm{Conj-Group} to isomorphism testing of permutation groups never appeared in print. As a result, we are not aware of any proof in the literature exhibiting a (polynomial-time) reduction from:
\begin{itemize}
\item \algprobm{Conj-Group} to isomorphism testing of permutation groups, or
\item \algprobm{Graph Isomorphism} to isomorphism testing of permutation groups.
\end{itemize}

\noindent Both reductions are known to experts. In addition to technically filling this second gap in the literature, Theorem~\ref{thm:GpIPermutationGroups} shows hardness of isomorphism testing of Fitting-free permutation groups; isomorphism of Fitting-free  groups was not previously known to be $\textsf{GI}$-hard in the permutation group (nor any succinct) model \cite{WilsonLevetEmail}.
\end{remark}

\begin{remark}
Cannon and Holt \cite{CH03} exhibited an algorithm for isomorphism testing of Fitting-free groups specified by generating sequences of permutations. While their algorithm performs well in practice, they do not formally analyze the runtime complexity.  Theorem~\ref{thm:GpIPermutationGroups} provides the first formal evidence we are aware of that their algorithm is unlikely to run in polynomial-time in the worst case, and that efficient practical algorithms may essentially be the best for which we can hope.
\end{remark}

We now return to the multiplication table model. Determining hard instances of \algprobm{Group Isomorphism} remains a long-standing open problem. Class-$2$ $p$-groups of exponent $p$ have long been considered bottleneck instances of \algprobm{GpI} (see e.g., \cite{BCGQ, GrochowQiaoPgroups}). However, a formal reduction from \algprobm{GpI} to isomorphism testing of class $2$ $p$-groups of exponent $p$ (or any other proper sub-class of groups), remains elusive. Existing evidence \cite{BCGQ, BCQ, DietrichWilson, CGLWISSAC} indicates that groups whose orders are \textit{not} large prime powers admit efficient isomorphism tests.

In the absence of formal reductions, one could ask for a complexity measure $\mathcal{C}$ and classes of groups $\mathcal{X}, \mathcal{Y}$ such that $\mathcal{C}(\mathcal{X}) < \mathcal{C}(\mathcal{Y})$. If we take our complexity measure $\mathcal{C}$ to be the minimum number of variables needed in a first-order formula (without counting) that identifies a group from the given class uniquely up to isomorphism, then previous work \cite{GrochowLevetWL} established that there are infinite families of Abelian groups that are not identified by \textsf{FO} formulas with $o(\log|G|)$ variables. In contrast, we show:

\begin{theorem} \label{thm:DescriptiveComplexityMain}
Every Fitting-free group of order $n$ is identified by a first-order formula using $O(\log \log n)$ variables.
\end{theorem}

Brachter \cite{BrachterThesis} previously showed the analogous result for $\textsf{FO} + \textsf{C}$ (first-order logic \textit{with} counting quantifiers); our result strengthens his by removing the need for counting.

Despite isomorphism testing of Abelian groups in the Cayley table model being vastly simpler for serial or parallel algorithms than isomorphism testing of Fitting-free groups, in more succinct input models the situation is less clear. However, the fact that Fitting-free groups are more easily identified by first-order logic \emph{without counting} than Abelian groups accords well with the idea that when dealing with isomorphism of groups containing Abelian normal subgroups, one must work with automorphism groups that are subgroups of matrix groups \cite{GR16}, which clearly ``require counting'' and are notoriously harder to work with than the more combinatorial subgroups of permutation groups that arise in testing isomorphism of Fitting-free groups \cite{BCGQ, BCQ} (see also \cite[Sec. 4, ``Non-Abelian linear algebra'']{Luks86}).

It is interesting to note that isomorphism testing of Abelian groups can be done using only a single counting gate \cite{ChattopadhyayToranWagner}, and the same was later established for other classes of groups \cite{CollinsLevetWL, CollinsUndergradThesis}, raising the question of a possible complexity hierarchy in the setting of first-order logic based on the number of counting quantifiers used (fixing the total number of variables and/or quantifiers).  We also note that, despite the close relationship between logics and circuits (cf. \cite{VollmerText, GroheVerbitsky}), establishing a circuit complexity analogue of Theorem~\ref{thm:DescriptiveComplexityMain} seems out of reach for current techniques.

\noindent \\ \textbf{Methods.} In order to prove Theorem~\ref{thm:main}, we take advantage of the fact that $G/\text{PKer}(G)$ can be viewed as a permutation group of degree $O(\log |G|)$, acting on $\Soc(G)$, in tandem with a standard suite of $\textsf{NC}$ algorithms for permutation groups \cite{BabaiLuksSeress}. As our permutation degree is $O(\log |G|)$ and $G$ is specified by its multiplication table, we can implement solutions for each subroutine of \cite{BabaiLuksSeress} using an $\textsf{AC}$ circuit of depth $\poly(\log \log |G|))$ ($\FOPLL$). We bring this suite of $\FOPLL$ permutation group algorithms to bear, in order to carefully analyze Luks's parallel algorithm for \algprobm{Coset Intersection} \cite{LuksHypergraphIso}, as well as to efficiently parallelize both the requisite instance of \algprobm{Twisted Code Equivalence} and the reduction to \algprobm{Twisted Code Equivalence} \cite{BCQ, CodenottiThesis}. 

Our analysis is specific to the algorithm at hand. If solutions to these problems used polynomial space and simply-exponential time, just scaling down to permutation groups on domain size $O(\log |G|)$ and using the standard simulation of space by depth would only yield $\poly(\log |G|)$-depth circuits of quasi-polynomial size. In contrast, by carefully analyzing Luks's \algprobm{Coset Intersection} procedure, we are able to get $\mathrm{poly}(\log \log |G|)$ depth and $\poly(|G|)$ size. Similarly, by carefully analyzing the existing \algprobm{Twisted Code Equivalence} algorithm and the corresponding reduction to \algprobm{Twisted Code Equivalence} \cite{BCQ, CodenottiThesis}, we obtain $O(\log^3 |G|)$ depth and size $\poly(|G|)$. 

Our strategy of leveraging a \textit{small} (size $O(\log |G|)$) permutation domain in order to get efficient instances of permutation group algorithms has been a common theme in \algprobm{Group Isomorphism} to obtain polynomial-time isomorphism tests. Crucially, the corresponding \textit{small} instances of \textsf{GI}-hard problems such as \algprobm{Linear Code Equivalence}, \algprobm{Twisted Code Equivalence}, \algprobm{Coset Intersection}, and \algprobm{Set Transporter} can be implemented in time $\poly(|G|)$---see, e.\,g., \cite{BMWGenus2, LewisWilson, IvanyosQ19, QST11, BQ, GQ15, BCGQ, BCQ, GQCoho}. By taking advantage of both a small permutation degree and existing parallel permutation algorithms \cite{BabaiLuksSeress}, we are able to improve the complexity-theoretic upper-bounds for isomorphism testing of Fitting-free groups from $\textsf{P}$ \cite{BCQ} to $\textsf{AC}^{3}$.

\noindent \\ \textbf{Further Related Work.} The work on $\textsf{NC}$ algorithms for \algprobm{GpI} is comparatively nascent compared to that of \algprobm{GI}. Indeed, much of the work involves parallelizing the generator-enumeration strategy (\emph{ibid}), which has yielded bounds of $\textsf{L}$ \cite{TangThesis} for $O(1)$-generated groups. Other families of groups known to admit $\textsf{NC}$ isomorphism tests include Abelian groups \cite{ChattopadhyayToranWagner, GrochowLevetWL, CGLWISSAC}, graphical groups arising from the CFI graphs \cite{WLGroups, CollinsLevetWL, CollinsUndergradThesis}, coprime extensions $H \ltimes N$ where $H$ is $O(1)$-generated and $N$ is Abelian \cite{GrochowLevetWL} (parallelizing a result from \cite{QST11}), groups of almost all orders \cite{CGLWISSAC} (parallelizing \cite{DietrichWilson}), and Fitting-free groups where the number of non-Abelian simple factors of the socle is $O(\log n/\log \log n)$ \cite{GrochowLevetWL} (parallelizing a result from \cite{BCGQ}). Our Theorem~\ref{thm:main} is the parallel analogue of \cite{BCQ}.

\section{Preliminaries}

\subsection{Groups and Codes}
\noindent \textbf{Groups.} The only perhaps unfamiliar terminology we use for groups is the term ``subcoset'' (which is standard in the relevant literature, but perhaps not standard in group theory textbooks). A \emph{subcoset} of a group $G$ is either the empty set or a coset of a subgroup of $G$.

\noindent \\ \textbf{Codes.} Let $\mathbb{F}$ be a field. $\text{GL}_{m}(\mathbb{F})$ denotes the set of $m \times m$ invertible matrices over $\mathbb{F}$. A \textit{linear code} of length $m$ is a subspace $U \leq \mathbb{F}^{m}$. A $d \times m$ matrix $A$ over $\mathbb{F}$ \textit{generates} the code $U$ if the rows of $A$ span $U$. Let $U,W$ be $d$--dimensional codes of length $m$ over $\mathbb{F}$, generated by $d \times m$ matrices $A, B$, respectively. Then $U$ and $V$ are \emph{equivalent} if and only there exists a permutation matrix $P \in \text{GL}_{m}(\mathbb{F})$ and a matrix $T \in \text{GL}_{d}(\mathbb{F})$ such that $B = TAP$.

\subsection{Computational Complexity} \label{sec:Complexity}
We assume that the reader is familiar with standard complexity classes such as $\textsf{P}, \textsf{NP}, \textsf{L}$, and $\textsf{NL}$. For a standard reference on circuit complexity, see \cite{VollmerText}. We consider Boolean circuits using the gates \textsf{AND}, \textsf{OR}, \textsf{NOT}, and \textsf{Majority}, where $\textsf{Majority}(x_{1}, \ldots, x_{n}) = 1$ if and only if $\geq n/2$ of the inputs are $1$. Otherwise, $\textsf{Majority}(x_{1}, \ldots, x_{n}) = 0$. In this paper, we will consider $\textsf{DLOGTIME}$-uniform circuit families $(C_{n})_{n \in \mathbb{N}}$. For this,
one encodes the gates of each circuit $C_n$ by bit strings of length $O(\log n)$. Then the circuit family $(C_n)_{n \geq 0}$
is called \emph{\textsf{DLOGTIME}-uniform}  if (i) there exists a deterministic Turing machine that computes for a given gate $u \in \{0,1\}^*$
of $C_n$ ($|u| \in O(\log n)$) in time $O(\log n)$ the type of gate $u$, where the types are $x_1, \ldots, x_n$, \textsf{NOT}, \textsf{AND}, \textsf{OR}, or \textsf{Majority} gates,
and (ii) there exists a deterministic Turing machine that decides for two given gates $u,v \in \{0,1\}^*$ of $C_n$ ($|u|, |v| \in O(\log n)$) and a binary encoded integer $i$ with $O(\log n)$ many bits in time $O(\log n)$ whether $u$ is the $i$-th input gate for $v$.

\begin{definition}
Fix $k \geq 0$. We say that a language $L$ belongs to (uniform) $\textsf{NC}^{k}$ if there exist a (uniform) family of circuits $(C_{n})_{n \in \mathbb{N}}$ over the $\textsf{AND}, \textsf{OR}, \textsf{NOT}$ gates such that the following hold:
\begin{itemize}
\item The $\textsf{AND}$ and $\textsf{OR}$ gates take exactly $2$ inputs. That is, they have fan-in $2$.
\item $C_{n}$ has depth $O(\log^{k} n)$ and uses (has size) $n^{O(1)}$ gates. Here, the implicit constants in the circuit depth and size depend only on $L$.

\item $x \in L$ if and only if $C_{|x|}(x) = 1$. 
\end{itemize}
\end{definition}

\noindent The complexity class $\AC^{k}$ is defined analogously as $\textsf{NC}^{k}$, except that the $\textsf{AND}, \textsf{OR}$ gates are permitted to have unbounded fan-in.
That is, a single $\textsf{AND}$ gate can compute an arbitrary conjunction, and a single $\textsf{OR}$ gate can compute an arbitrary disjunction. The class $\SAC^k$ is defined analogously, in which the $\textsf{OR}$ gates have unbounded fan-in but the $\textsf{AND}$ gates must have fan-in $2$.
The complexity class $\textsf{TC}^{k}$ is defined analogously as $\AC^{k}$, except that our circuits are now permitted $\textsf{Majority}$ gates of unbounded fan-in.
We also allow circuits to compute functions by using multiple output gates. 

For every $k$, the following containments are well-known:
\[
\textsf{NC}^{k} \subseteq \SAC^k \subseteq  \AC^{k} \subseteq \textsf{TC}^{k} \subseteq \textsf{NC}^{k+1}.
\]

\noindent In the case of $k = 0$, we have that:
\[
\textsf{NC}^{0} \subsetneq \AC^{0} \subsetneq \textsf{TC}^{0} \subseteq \textsf{NC}^{1} \subseteq \LogSpace \subseteq \textsf{NL} \subseteq \textsf{SAC}^{1} \subseteq \AC^{1}.
\]

\noindent We note that functions that are $\textsf{NC}^{0}$-computable can only depend on a bounded number of input bits. Thus, $\textsf{NC}^{0}$ is unable to compute the $\textsf{AND}$ function. It is a classical result that $\AC^{0}$ is unable to compute \algprobm{Parity} \cite{FSS}. The containment $\textsf{TC}^{0} \subseteq \textsf{NC}^{1}$ (and hence, $\textsf{TC}^{k} \subseteq \textsf{NC}^{k+1}$) follows from the fact that $\textsf{NC}^{1}$ can simulate the unbounded fan-in \textsf{Majority} gate.

The complexity class $\textsf{FOLL}$ is the set of languages decidable by uniform $\textsf{AC}$ circuits of depth $O(\log \log n)$ and polynomial-size \cite{BKLM}. We will use a slight generalization of this: we use $\FOPLL$ to denote the class of languages $L$ decidable by uniform $\textsf{AC}$ circuits of depth $O((\log \log n)^c)$ for some $c$ that depends only on $L$. It is known that $\textsf{AC}^{0} \subsetneq \FOPLL \subsetneq \textsf{AC}^{1}$, the former by a simple diagonalization argument on top of Sipser's result \cite{SipserBorel}, and the latter because the \textsf{Parity} function is in $\mathsf{AC}^1$ but not $\FOPLL$ (nor any depth $o(\log n / \log \log n)$). \FOPLL\ cannot contain any complexity class that can compute \algprobm{Parity}, such as $\mathsf{TC}^0, \mathsf{NC}^1, \mathsf{L}, \mathsf{NL}$, or $\mathsf{SAC}^1$, and it remains open whether any of these classes contain \FOPLL.

\subsection{Permutation Group Algorithms}

We will consider the permutation group model, in which groups are specified succinctly by a sequence of permutations from $S_{m}$. The computational complexity for the permutation group model will be measured in terms of $m$. We recall a standard suite of problems with known $\textsf{NC}$ solutions in the setting of permutation groups. In our setting, $m$ will be \textit{small} ($m \in O(\text{polylog}(n))$) relative to the overall input size $n$, which will yield $\FOPLL$ bounds for these problems.

\begin{lemma} \label{PermutationGroupsNC}
Let $c \in \mathbb{Z}^{+}$, and let $m \in O(\log^c n)$. Let $G \leq S_{m}$ be given by a sequence $S$ of generators. The following problems are in $\FOPLL$ relative to $n$, that is, they have uniform $\mathsf{AC}$ circuits of depth $\poly(\log \log n)$ and $\poly(n)$ size:
\begin{enumerate}[label=(\alph*)]
\item Compute the order of $G$.
\item Decide whether a given permutation $\sigma$ is in $G$; and if so, exhibit a word $\omega$ such that $\sigma = \omega(S)$. 

\item Find the kernel of any action of $G$.
\item Find the pointwise stabilizer of $B \subseteq [m]$. 
\item Find the normal closure of any subset of $G$.

\item Compute $\FOPLL$-efficient Schreier generators. In particular, we may take the corresponding chain of subgroups $G = G_0 \geq G_1 \geq \cdots \geq G_r = 1$ such that $r \in O(\log n)$ and $G_{i}$ is the point-stabilizer of the first $i$ points. 

\item Given a list of generators $S$ for $G$, compute a minimal (non-redundant) of generators $S'$ for $G$ of size $\poly(m)$. 
\item If $H \leq S_{m}$ and we are given a bijection $f : \text{dom}(G) \to \text{dom}(H)$, decide whether $f$ is a permutational isomorphism of $G$ and $H$.
\end{enumerate}
\end{lemma}

\begin{proof}
Each of these problems is known to be solvable using a circuit of depth $O(\text{polylog}(m))$ and size $O(\poly(m))$: see \cite{BabaiLuksSeress} for (a)--(g). For (h), we apply (b) to check membership of the $f$-images of the generators of $G$ in $H$, and vice-versa. The $\FOPLL$ bound follows from the fact that $m \in O(\log^c n)$. 
\end{proof}

\subsection{Weisfeiler--Leman} \label{sec:WLPrelims}

We begin by recalling the Weisfeiler--Leman algorithm for graphs, which computes an isomorphism-invariant coloring. Let $\Gamma$ be a graph, and let $k \geq 2$ be an integer. The $k$-dimensional Weisfeiler--Leman, or $k$-WL, algorithm begins by constructing an initial coloring $\chi_{0} : V(\Gamma)^{k} \to \mathcal{K}$, where $\mathcal{K}$ is our set of colors, by assigning each $k$-tuple a color based on its (marked) isomorphism type. That is, two $k$-tuples $(v_{1}, \ldots, v_{k})$ and $(u_{1}, \ldots, u_{k})$ receive the same color under $\chi_{0}$ iff the map $v_i \mapsto u_i$ (for all $i \in [k]$) is an isomorphism of the induced subgraphs $\Gamma[\{ v_{1}, \ldots, v_{k}\}]$ and $\Gamma[\{u_{1}, \ldots, u_{k}\}]$ and for all $i, j$, $v_i = v_j \Leftrightarrow u_i = u_j$. 

For $r \geq 0$, the coloring computed at the $r$th iteration of Weisfeiler--Leman is refined as follows. For a $k$-tuple $\overline{v} = (v_{1}, \ldots, v_{k})$ and a vertex $x \in V(\Gamma)$, define
\[
\overline{v}(v_{i}/x) = (v_{1}, \ldots, v_{i-1}, x, v_{i+1}, \ldots, v_{k}).
\]
The coloring computed at the $(r+1)$st iteration, denoted $\chi_{r+1}$, stores the color of the given $k$-tuple $\overline{v}$ at the $r$th iteration, as well as the colors under $\chi_{r}$ of the $k$-tuples obtained by substituting a single vertex in $\overline{v}$ for another vertex $x$. We examine this multiset of colors over all such vertices $x$. This is formalized as follows:
\begin{align*}
\chi_{r+1}(\overline{v}) = &( \chi_{r}(\overline{v}), \{\!\!\{ ( \chi_{r}(\overline{v}(v_{1}/x)), \ldots, \chi_{r}(\overline{v}(v_{k}/x) ) \bigr| x \in V(\Gamma) \}\!\!\} ),
\end{align*}
where $\{\!\!\{ \cdot \}\!\!\}$ denotes a multiset.

The \textit{count-free} variant of WL considers the set rather than the multiset of colors at each round. Precisely:
\begin{align*}
\chi_{r+1}(\overline{v}) = &( \chi_{r}(\overline{v}), \{ ( \chi_{r}(\overline{v}(v_{1}/x)), \ldots, \chi_{r}(\overline{v}(v_{k}/x) ) \bigr| x \in V(\Gamma) \} ).
\end{align*}

\noindent Note that the coloring $\chi_{r}$ computed at iteration $r$ induces a partition of $V(\Gamma)^{k}$ into color classes. The Weisfeiler--Leman algorithm terminates when this partition is not refined, that is, when the partition induced by $\chi_{r+1}$ is identical to that induced by $\chi_{r}$. The final coloring is referred to as the \textit{stable coloring}, which we denote $\chi_{\infty} := \chi_{r}$.

Brachter and Schweitzer introduced three variants of WL for groups \cite{WLGroups}. We will restrict attention to WL Version I, which we will now recall. WL Version I is executed directly on the groups, where $k$-tuples of group elements are initially colored. Two $k$-tuples $(g_{1}, \ldots, g_{k})$ and $(h_{1}, \ldots, h_{k})$ receive the same initial color iff (a) for all $i, j, \ell \in [k]$, $g_{i}g_{j} = g_{\ell} \iff h_{i}h_{j} = h_{\ell}$, and (b) for all $i, j \in [k]$, $g_{i} = g_{j} \iff h_{i} = h_{j}$. Refinement is performed in the classical manner as for graphs. Namely, for a given $k$-tuple $\overline{g}$ of group elements,
\begin{align*}
\chi_{r+1}(\overline{g}) = &( \chi_{r}(\overline{g}), \{\!\!\{ ( \chi_{r}(\overline{g}(g_{1}/x)), \ldots, \chi_{r}(\overline{g}(g_{k}/x) ) \bigr| x \in G \}\!\!\} ).
\end{align*}

Grohe \& Verbitsky \cite{GroheVerbitsky} previously showed that for fixed $k$, the classical $k$-dimensional Weisfeiler--Leman algorithm for graphs can be effectively parallelized. Precisely, each iteration of the classical counting WL algorithm (including the initial coloring) can be implemented using a logspace uniform $\textsf{TC}^{0}$ circuit, and each iteration of the \textit{count-free} WL algorithm can be implemented using a logspace uniform $\textsf{AC}^{0}$ circuit. A careful analysis shows that this parallelization holds under $\textsf{DLOGTIME}$-uniformity. As they mention (\cite[Remark~3.4]{GroheVerbitsky}), their implementation works for any first-order structure, including groups. In particular, the initial coloring of count-free WL Version I is $\textsf{DLOGTIME}$-uniform $\textsf{AC}^{0}$-computable, and each refinement step is $\textsf{DLOGTIME}$-uniform $\textsf{AC}^{0}$-computable.

\subsection{Pebbling Game}

The classical counting Weisfeiler--Leman algorithm is equivalent to classical Ehrenfeucht--Fra\"iss\'e bijective pebble games \cite{Hella1989, Hella1993, ImmermanLander1990, CFI}. These games are often used to show that two graphs $X$ and $Y$ cannot be distinguished by $k$-WL. There exists an analogous pebble game for count-free WL \cite{ImmermanLander1990, CFI}. We will restrict attention to the count-free WL Version I pebble game for groups. The count-free $(k+1)$-pebble game consists of two players: Spoiler and Duplicator, as well as $(k+1)$ pebble pairs $(p, p^{\prime})$. Spoiler wishes to show that the two groups $G$ and $H$ are not isomorphic. Duplicator wishes to show that the two groups are isomorphic. Each round of the game proceeds as follows.
\begin{enumerate}
\item Spoiler picks up a pebble pair $(p_{i}, p_{i}^{\prime})$.
\item Spoiler places one of the pebbles on some group element (either $p_{i}$ on some element of $G$ or $p_{i}'$ on some element of $H$). 
\item Duplicator places the other pebble on some element of the other group.
\item The winning condition is checked. This will be formalized later.
\end{enumerate}

Let $v_{1}, \ldots, v_{m}$ be the pebbled elements of $G$ at the end of step 1, and let $v_{1}^{\prime}, \ldots, v_{m}^{\prime}$ be the corresponding pebbled vertices of $H$. Spoiler wins precisely if the map $v_{\ell} \mapsto v_{\ell}^{\prime}$ does not satisfy both of the following: (a) $v_{i}v_{j} = v_{k} \iff v_{i}'v_{j}' = v_{k}'$ for all $i, j, k \in [m]$, and (b) $v_{i} = v_{j} \iff v_{i}' = v_{j}'$ for all $i, j \in [m]$. Duplicator wins otherwise. Spoiler wins, by definition, at round $0$ if $G$ and $H$ do not have the same number of elements. Two $k$-tuples $\bar{u} \in G^k, \bar{v} \in H^k$ are not distinguished by the first $r$ rounds of the count-free $k$-dimensional WL Version I, if and only if Duplicator wins the first $r$ rounds of the Version I count-free $(k+1)$-pebble game \cite{ImmermanLander1990, CFI, GrochowLevetWL}.

\subsection{Logics} \label{sec:Logics}

We recall the central aspects of first-order logic. We have a countable set of variables $\{x_{1}, x_{2}, \ldots, \}$. Formulas are defined inductively. As our basis, $x_{i} = x_{j}$ is a formula for all pairs of variables. Now if $\varphi, \psi$ are formulas, then so are the following: $\varphi \land \psi, \varphi \vee \psi, \neg{\varphi}, \exists{x_{i}} \, \varphi,$ and $\forall{x_{i}} \, \varphi$. In order to define logics on groups, it is necessary to define a relation that relates the group multiplication. We recall the Version I logic introduced by Brachter \& Schweitzer \cite{WLGroups}. Here, we add a ternary relation $R$ where $R(x_{i}, x_{j}, x_{\ell}) = 1$ if and only if $x_{i}x_{j} = x_{\ell}$ in the group. In keeping with the conventions of \cite{CFI}, we refer to the first-order logic with relation $R$ as $\mathcal{L}^{I}$ and its $k$-variable fragment as $\mathcal{L}^{I}_{k}$. If furthermore we restrict the formulas to have quantifier depth at most $r$, we denote this fragment as $\mathcal{L}^{I}_{k,r}$. Given a formula $\varphi$ with $k$ free (=unquantified) variables, and a $k$-tuple $\bar{g} \in G^k$, we say that $(G,\bar{g})$ \emph{satisfies} or \emph{is a model of} $\varphi$, denoted $(G, \bar{g}) \models \varphi$, if, after plugging in $g_1,\dotsc,g_k$ into the unquantified variables, the resulting formula is true in $G$. Two $k$-tuples $\bar{g} \in G^k, \bar{h} \in H^k$ are distinguished by the first $r$ rounds of the count-free $k$-WL Version I if and only if there exists a formula $\varphi \in \mathcal{L}^{I}_{k+1,r}$ such that $(G, \bar{g}) \models \varphi$ and $(H, \bar{h}) \not \models \varphi$ (see e.g., \cite{ImmermanLander1990, CFI, GrochowLevetWL}).

\section{Coset Intersection for small-domain groups in parallel}
In this section, we consider the \algprobm{Coset Intersection} problem, which takes $G, H \leq S_{m}$ and $x, y \in S_{m}$, and asks for $Gx \cap Hy$. Precisely, we will establish the following.

\begin{proposition} \label{prop:CosetIntersection}
Fix $n \in \mathbb{Z}^{+}$, and let $m \in O(\log n)$. Let $G, H \leq S_{m}$ and $x, y \in S_{m}$. We can compute $Gx \cap Hy$ using an $\textsf{AC}$ circuit of depth $\poly(\log \log n)$ and size $\poly(n)$.
\end{proposition}

To prove Proposition~\ref{prop:CosetIntersection}, we carefully analyze Luks' parallel procedure \cite{LuksHypergraphIso}. Our main contribution is to precisely determine the circuit depth. We do so by taking advantage of the fact that we have a \textit{small} permutation group, in tandem with a suite of existing parallel algorithms for permutation groups (Lemma~\ref{PermutationGroupsNC}). Our analysis is specific to the algorithm at hand; if Luks's algorithm used polynomial space and simply-exponential time, just scaling down to permutation groups on domain size $O(\log n)$ and using the standard simulation of space by depth would only yield $\poly(\log n)$-depth circuits of quasi-polynomial size. In contrast, by carefully analyzing Luks's algorithm we are able to get $\mathrm{poly}(\log \log n)$ depth and $\poly(n)$ size. We now turn to the details.

\begin{proof}
Luks \cite[Corollary~3.2]{LuksHypergraphIso} showed that \algprobm{Coset Intersection} reduces to the following problem, which he calls \algprobm{Problem I}. His reduction is short and readily seen to be implementable by a linear-size, constant-depth circuit.

\begin{definition}
Luks's \algprobm{Problem I} is defined as follows:
\begin{itemize}
    \item \textsf{Instance:} $L \leq \text{Sym}(\Gamma) \times\text{Sym}(\Delta)$, $z \in \text{Sym}(\Gamma \times \Delta)$, and $\Pi \subseteq \Gamma \times \Delta$.
    \item \textsf{Solution:} $(Lz)_{\Pi} = \{ x \in Lz : \Pi^{x} = \Pi\}$.
\end{itemize}
\end{definition}

Thus, it suffices to show that \algprobm{Problem I} can be solved in the stated depth and size. This is the content of the following Proposition~\ref{prop:CosetIntersection2}, which will complete the proof.
\end{proof}

\noindent Luks \cite[Proposition~6.1]{LuksHypergraphIso} showed that \algprobm{Problem I} admits a parallel solution; taking his Prop.~6.1 as a black box and scaling down the parameters to $\log n$ would yield circuits of $\poly(n)$ size and $\poly(\log n)$-depth. By careful analysis of the depth complexity of his algorithm we will show the stronger result that when $|\Gamma|, |\Delta| \in O(\log n)$, we can solve \algprobm{Problem I} using an $\textsf{AC}$ circuit of depth $\poly(\log \log n)$ and size $\poly(n)$. 

\begin{proposition}[cf. {\cite[Proposition~6.1]{LuksHypergraphIso}}] \label{prop:CosetIntersection2}
Suppose that $|\Gamma|, |\Delta| \in O(\log n)$. Then \algprobm{Problem I} can be solved using an $\textsf{AC}$ circuit of depth $\poly(\log \log n)$ and size $\text{poly}(n)$.
\end{proposition}

\begin{proof}
We carefully analyze \cite[Proposition~6.1]{LuksHypergraphIso}. Without loss of generality, we may assume that $|\Gamma|, |\Delta|$ are powers of $2$. Otherwise, we can augment $\Gamma \subseteq \Gamma', \Delta \subseteq \Delta'$, and then do the following:
\begin{itemize}
\item let $z$ act trivially on $(\Gamma' \times \Delta') \setminus (\Gamma \times \Delta)$ and, 
\item for $(x,y) \in L$, let $x$ act trivially on $\Gamma' \setminus \Gamma$ and $y$ act trivially on $\Delta' \setminus \Delta$. 
\end{itemize}

\noindent As $\Pi \subseteq \Gamma \times \Delta$, which is stabilized by $Lz$, augmenting $\Gamma$ and $\Delta$ does not change $(Lz)_{\Pi}$. 

To accommodate recursion, Luks introduced the following \algprobm{Problem II}.

\begin{definition}
Luks's \algprobm{Problem II} is defined as follows:
\begin{itemize}
    \item \textsf{Instance:} $L \leq \text{Sym}(\Gamma) \times \text{Sym}(\Delta)$, 
    $z \in \text{Sym}(\Gamma \times \Delta)$, 
    $\Pi \subseteq \Gamma \times \Delta$, and $\Theta = \Phi \times \Psi \subseteq \Gamma \times \Delta$ with $|\Theta|$ a power of $2$ such that $\text{Stab}_{L}(\Theta) = L$.

    \item \textsf{Solution:} $(Lz)_{\Pi}[\Theta] = \{ x \in Lz : (\Pi \cap \Theta)^x = \Pi \cap \Theta^x \}$.
\end{itemize}
\end{definition}

\noindent \algprobm{Problem I} is a special case in which $\Theta = \Gamma \times \Delta$. We now proceed to analyze the depth of Luks's algorithm for \algprobm{Problem II} in our small-domain setting.

Luks \cite[Proposition~3.1]{LuksHypergraphIso} established that the output to \algprobm{Problem II} will be either $\emptyset$ or a right coset of $L_{\Pi}[\Theta]$. Following Luks, as $\Pi$ is fixed throughout, we will write $(Ny)[\Theta]$ instead of $(Ny)_{\Pi}[\Theta]$ for any subcoset $Ny$ of $\text{Sym}(\Gamma \times \Delta)$. We have the following cases.
\begin{itemize}
\item \textbf{Case 1:} Suppose that $|\Pi \cap \Theta| \neq |\Pi \cap \Theta^{z}|$. These intersections can be computed by an $\mathsf{AC}$ circuit of depth 2 and $\poly(\log n)$ size; their sizes can then be compared using a single threshold gate with $\poly(\log n)$ many inputs, which can then be simulated by an $\mathsf{AC}$ circuit of depth $O(\log \log n)$ and size $\poly(\log n)$.  Then $(Lz)[\Theta] = \emptyset$ (note that $\Theta^{x} = \Theta^z$ for all $x \in Lz$). 

\item \textbf{Case 2:} Suppose that both $\Pi \cap \Theta$ and $\Pi \cap \Theta^z$ are empty. As in the previous case, these intersections can be computed by an $\mathsf{AC}$ circuit of depth 2 and $\poly(\log n)$ size; checking that they are empty then increases the depth by only 1. Then $(Lz)[\Theta] = Lz$. Thus, from this point forward, we may assume without loss of generality that $\Pi \cap \Theta$ and $\Pi \cap \Theta^z$ are non-empty and the same size.

\item \textbf{Case 3:} Suppose that $|\Pi \cap \Theta| = 1$. Luks established in the proof of \cite[Proposition~3.1]{LuksHypergraphIso} that the problem in this case reduces to computing a point-stabilizer. This is $\FOPLL$-computable, by Lemma~\ref{PermutationGroupsNC}. 

\item \textbf{Case 4:} Suppose that Cases 1--3 do not hold. Then we have that $|\Pi \cap \Theta| > 1$. If $|\Phi| > 1$, then fix some $\Phi_{1} \subsetneq \Phi$ with $|\Phi_{1}| = |\Phi|/2$, and let $\Phi_{2} := \Phi \setminus \Phi_{1}$. For $i \in [2]$, define $\Theta_{i} := \Phi_{i} \times \Psi$. If instead $|\Phi| = 1$, then fix some $\Psi_{1} \subsetneq \Psi$ with $\Psi_{1} = |\Psi|/2$, and let $\Psi_{2} := \Psi \setminus \Psi_{1}$. We then take, for $i \in [2]$, $\Theta_{i} := \Phi \times \Psi_{i}$.

Without loss of generality, suppose that $\Theta_{1} = \Phi_{1} \times \Psi$. Let $M = \text{Stab}_{L}(\Theta_{1})$. In the proof of \cite[Proposition~6.1]{LuksHypergraphIso}, Luks notes that it would be prohibitive to enumerate $2^{\Phi \times \Psi}$, and so we determine the orbit $\mathcal{O}$ of $\Phi_{1}$ only under the action of $L$ on the first coordinate. By Lemma~\ref{PermutationGroupsNC}(f), we can compute a set $S$ of strong generators for $M$, together with a transversal $T$ of $L/M$, in $\FOPLL$. Using $S$ and $T$, we can compute Schreier generators for $M$. We need to ensure that the number of Schreier generators is $\text{poly}(n)$. Luks \cite[Proposition~6.1]{LuksHypergraphIso} shows a stronger result that we can obtain $\text{poly}(|\Gamma \times \Delta|) = O(\text{polylog}(n))$ generators. This can be accomplished in $\FOPLL$ using Lemma~\ref{PermutationGroupsNC}(g). 

Now observe that:
\[
L = \bigcup_{t \in T} Mt.
\]

As $\Theta_{1}$ is stabilized by $M$, we have that $Mtz[\Theta_{1}]$ is a coset of $L/M$, provided $Mtz[\Theta_{1}]$ is non-empty. So we can proceed recursively for each $t$:
\[
(Mtz)[\Theta] = ((Mtz)[\Theta_{1}])[\Theta_{2}].
\]

Each of the $|T|$ recursive calls can be executed in parallel. The calls to problems on $\Theta_1$ and $\Theta_2$ have to be made serially. However, there are only two such problems, each of half the size of $\Theta$, so the recursion tree will look like a binary tree of height $O(\log |\Theta|) = O(\log \log n)$. 

Finally, observe that:
\[
(Lz)[\Theta] = \bigcup_{t \in T} (Mtz)[\Theta].
\]

In particular, each non-empty $Mtz[\Theta]$ is a coset of $M[\Theta]$, the subgroup of $M$ stabilizing $\Pi \cap \Theta$. Note that when combining cosets, we need to ensure that the number of group generators is $\text{poly}(n)$. To accomplish this, we use Lemma~\ref{PermutationGroupsNC}(g). As our permutation domain has size $O(\log^2 n)$, each application of Lemma~\ref{PermutationGroupsNC}(g) is $\FOPLL$-computable. We use a binary tree circuit of depth $\log |T| \in O(\log \log n)$ to compute $(Lz)[\Theta]$ from the $(Mtz)[\Theta]$ ($t \in T$). This yields a circuit of depth $\poly(\log \log n)$ to compute $(Lz)[\Theta]$.

Luks \cite[Proposition~6.1]{LuksHypergraphIso} showed that the recursion depth is $O(\log |\Phi| + \log |\Psi|) = O(\log \log n)$. We have shown that each recursive step is computable using an $\textsf{AC}$ circuit of depth $\poly(\log \log n)$ and size $\poly(n)$. Thus, our circuit has depth $\poly(\log \log n)$ and size $\text{poly}(n)$, as desired. \qedhere
\end{itemize}
\end{proof}

\section{Twisted Code Equivalence over small domains in parallel}

Throughout this section up to and including Section~\ref{sec:FF}, we frequently cite both the published extended abstract Babai--Codenotti--Qiao \cite{BCQ} and Codenotti's thesis \cite{CodenottiThesis}. As an extended abstract, \cite{BCQ} omitted some details of proofs, whereas \cite{CodenottiThesis} contains additional exposition and full details. 

We first recall the definition of \algprobm{Twisted Code Equivalence} from \cite{BCQ, CodenottiThesis}. 

\begin{definition}[{\cite{BCQ,CodenottiThesis}} Twisted equivalence of codes]
Let $\Gamma$ be an alphabet, and let $\mathcal{A} \subseteq \Gamma^{A}, \mathcal{B} \subseteq \Gamma^{B}$ be codes of length $n$ over $\Gamma$. Let $G$ be a group acting on $\Gamma$. Given a bijection $\pi : A \to B$ and a function $g : B \to G$, we say that $\psi = (\pi, g)$ is a $G$-\textit{twisted equivalence} of the codes $\mathcal{A}$ and $\mathcal{B}$ if $\mathcal{A}^{\psi} = \mathcal{B}$. We denote the set of all $G$-twisted equivalences of $\mathcal{A}$ and $\mathcal{B}$ as $\text{EQ}_{G}(\mathcal{A}, \mathcal{B})$. 

We will now generalize the notion of twisted equivalence to codes over multiple alphabets. Let $\Gamma_{1}, \ldots, \Gamma_{r}$ be our disjoint, finite alphabets. A string of \textit{length} $(k_1, \ldots, k_r)$ over $(\Gamma_1, \ldots, \Gamma_r)$ is a set of maps $x_i : A_i \to \Gamma_i$, denoted collectively as $x$, where $|A_i| = k$. A code of length $(k_1, \ldots, k_r)$ has total length $n := \sum_{i=1}^{r} x_i$. The set of all strings of length $(k_1, \ldots, k_r)$ over $(\Gamma_1, \ldots, \Gamma_r)$ is $\prod_{i=1}^{r} \Gamma_{i}^{A_{i}}$. A \textit{code} of length $(k_1, \ldots, k_r)$ with domain $(A_1, \ldots, A_r)$ is a subset $\mathcal{A} \subseteq \prod_{i=1}^{r} \Gamma_{i}^{A_{i}}.$

For each $i \in [r]$, let $G_i \leq \text{Sym}(\Gamma_i)$ be a group acting on $\Gamma_i$, and let $\mathcal{G} := (G_{1}, \ldots, G_{r})$. Given bijections $\pi_{i} : A_i \to B_i$ and functions $g : B_i \to G_i$ ($i \in [r]$), let $\psi_{i} := (\pi_i, g_i)$ and $\psi := (\psi_1, \ldots, \psi_r)$. For $a = (a_1, \ldots, a_r) \in \prod_{i=1}^{r} \Gamma_{i}^{A_{i}}$, define $a^{\psi} := (a_{1}^{\psi_{1}}, \ldots, a_{r}^{\psi_{r}})$. Now a $\mathcal{G}$-\textit{twisted equivalence} of two codes $\mathcal{A} \subseteq \prod_{i=1}^{r} \Gamma_{i}^{A_{i}}$ and $\mathcal{B} \subseteq \prod_{i=1}^{r} \Gamma_{i}^{B_{i}}$ is a $\psi$ such that $\mathcal{A}^{\psi} = \mathcal{B}$. The set of all $\mathcal{G}$-twisted equivalences of $\mathcal{A} \to \mathcal{B}$ is denoted $\text{EQ}_{\mathcal{G}}(\mathcal{A}, \mathcal{B})$. 
\end{definition}

\begin{definition}[{\cite{BCQ}}]
The \algprobm{Twisted Code Equivalence} problem takes as input two codes $\mathcal{A} \subseteq \prod_{i=1}^{r} \Gamma_{i}^{A_{i}}$ and $\mathcal{B} \subseteq \prod_{i=1}^{r} \Gamma_{i}^{B_{i}}$, as well as a sequence of groups $\mathcal{G} = (G_1, \ldots, G_r)$, with $G_{i} \leq \text{Sym}(\Gamma_i)$ given by a faithful permutation representation, and decides as output whether $\text{EQ}_{\mathcal{G}}(\mathcal{A}, \mathcal{B}) \neq \emptyset$.
\end{definition}

\begin{theorem}[cf. {\cite[Thm.~5]{BCQ}=\cite[Thm.~4.2.1]{CodenottiThesis}}] \label{thm:TwistedCodeEquivalence}
Let $m \in O(\log n)$. Let $\mathcal{A} \subseteq \prod_{i=1}^{r} \Gamma_{i}^{A_{i}}$ and $\mathcal{B} \subseteq \prod_{i=1}^{r} \Gamma_{i}^{B_{i}}$ be codes of total length $m$. For each $i \in [r]$, let $G_i$ be a group acting on $\Gamma_i$ given by its multiplication table.\footnote{The formulation in \cite{BCQ, CodenottiThesis} assumes that the $G_i$ are given by faithful permutation representations. For isomorphism testing of Fitting-free groups, we will be given the $G_i$ by their multiplication tables.} Let $G_{\max}$ and $\Gamma_{\max}$ such that $|G_{\max}| = \max_{i \in [r]} |G_{i}|$, and $|\Gamma_{\max}| = \max_{i \in [r]} |\Gamma_i|$. Then $\text{EQ}_{(G_1, \ldots, G_r)}(\mathcal{A}, \mathcal{B})$ can be found using an $\textsf{AC}$ circuit of depth  $O( (\log n) \cdot (\log |\Gamma_{\max}|) \cdot \poly(\log \log n))$ and size $\poly(n)$.
\end{theorem}

\begin{proof}
We follow the details of \cite[Theorem~4.2.1]{CodenottiThesis}. For subsets $U_{i} \subseteq A_{i}$, we call the functions $y : \bigcup_{i=1}^{r} U_{i} \to \bigcup_{i=1}^{r} \Gamma_{i}$ mapping $U_{i}$ into $\Gamma_{i}$ the \textit{partial strings} over $A = (A_1, \ldots, A_r)$, and we call the tuple $(|U_1|, \ldots, |U_r|)$ the \textit{length} of $y$. For every partial string $y$ over $A$, let $\mathcal{A}_{y}$ be the set of strings in $\mathcal{A}$ that are extensions of $y$. The analogous definitions hold for $\mathcal{B}$. Let $\mathcal{G} = (G_1, \ldots, G_r)$.

We construct a dynamic programming table with an entry for each pair $(y, z)$ of partial strings, $y$ over $A$ and $z$ over $B$, such that $\mathcal{A}_{y} \neq \emptyset$ and $B_{z} \neq \emptyset$. For each such pair, we store the set $I(y, z)$ of twisted equivalences of the restriction of $\mathcal{A}_{y}$ to $A \setminus \text{dom}(y)$ with the restriction of $\mathcal{B}_{Z}$ to $B \setminus \text{dom}(z)$. The sets $I(y,z)$ are either empty or cosets, and so we store them by generators and a coset representative. Note  that the condition $\mathcal{A}_{y} \neq \emptyset$ and $\mathcal{B}_{y} \neq \emptyset$ implies that only partial strings actually appearing in $\mathcal{A}$ and $\mathcal{B}$ will be recorded.

We proceed in $m$ stages. We will show that each stage can be implemented using an $\textsf{AC}$ circuit of depth $O( (\log |\Gamma_{\max}|) \cdot \poly(\log \log n))$ and size $\poly(n)$. As $m \in O(\log n)$, we obtain an $\textsf{AC}$ circuit of depth 
$O((\log n) \cdot (\log |\Gamma_{\max}|) \cdot \poly(\log \log n))$ and size $\text{poly}(n)$.

At the initial stage, we consider each pair $(y, z)$ of full strings, with $y$ over $\mathcal{A}$ and $z$ over $\mathcal{B}$. Our final stage will involve considering $\text{dom}(y) = \text{dom}(z) = \emptyset$; at which point, we will have constructed all $\mathcal{A} \to \mathcal{B}$ twisted $\mathcal{G}$-equivalences. When $y, z$ are full strings, we have $|\mathcal{A}_{y}| = 1$ and $|\mathcal{B}_{z}| = 1$. So $|I(y, z)| \leq 1$, and thus $I(y, z)$ is trivial to compute. It follows that we can compute in $\textsf{AC}^{0}$ all such $I(y,z)$, where again $y, z$ are full strings. 

Now fix $0 < \ell \leq m$, and suppose that for all $y', z'$ of total length greater than $\ell$ that we have computed $I(y', z')$. We will consider, in parallel, all pairs of partial strings $(y, z)$ of length $\ell$. Fix such $y, z$. To construct $I(y, z)$, we augment the domain of $y$ by one index $t \in A$, and the domain of $z$ by one index $s \in B$. For each $i \in [r]$, we proceed as follows in parallel. Fix  $t \in A_i$, and consider in parallel all possible choices of $s \in B_i$. For each $s \in B_i$ and $g \in G_i$, we will find the strings in $I(y, z)$ that move $s \to t$ and act on the symbol in position $\gamma$ by $g$. Precisely, let $y(t, \gamma)$ be the partial string extending $y$ by $\gamma$ at position $t$, and define $z(s, \gamma^{g^{-1}})$ analogously. We have that:
\[
I(y, z) = \bigcup_{s \in B_i} \bigcup_{g \in G_i} \bigcap_{ \substack{\gamma \in \Gamma_i \\ \mathcal{A}_{y(t, \gamma)} \neq \emptyset}} I( y(t, \gamma), z(s, \gamma^{g^{-1}})).
\]

\noindent If $z(s, \gamma^{g^{-1}}) \in \mathcal{B}$, then we can look up the value of $I(y(t,\gamma), z(s, \gamma^{g^{-1}}))$ in the table. If for some $\gamma, g$, $z(s, \gamma^{g^{-1}}) \not \in \mathcal{B}$, then $I(y(t,\gamma), z(s, \gamma^{g^{-1}})) = \emptyset$.

We may compute:
\[
\bigcap_{ \substack{\gamma \in \Gamma_i \\ \mathcal{A}_{y(t, \gamma)} \neq \emptyset}} I( y(t, \gamma), z(s, \gamma^{g^{-1}}))
\]
using a binary tree circuit of depth $O(\log |\Gamma_i|)$. At each node of the binary tree, we use a call to \algprobm{Coset Intersection}. As our permutation domain is of size $O(\log n)$, we have by Proposition~\ref{prop:CosetIntersection} that each call to \algprobm{Coset Intersection} is computable using an $\textsf{AC}$ circuit of depth $\poly(\log \log n)$ and size $\poly(n)$. This yields an $\textsf{AC}$ circuit of depth $O(\log(|\Gamma_i|) \cdot \poly(\log \log n))$.

We may now compute $I(y, z)$ by taking a union over all $s \in B_i$ and all $g \in G_i$. Now we must take care to keep the number of generators to be at most $\poly(n)$. We can  accomplish this with Lemma~\ref{PermutationGroupsNC}(g). As our permutation domain has size $O(\log n)$, each application of Lemma~\ref{PermutationGroupsNC}(g) is $\FOPLL$-computable. Now we use a binary tree circuit of depth $O(\log |B_i| + \log |\Gamma_i|) \subseteq O(\log \log n + \log |\Gamma_{\max}|)$, where at each node, we merge two cosets and then apply an instance of Lemma~\ref{PermutationGroupsNC}(g). Thus, taking the union can be accomplished using an $\textsf{AC}$ circuit of depth $O(\poly(\log \log n) + (\log |\Gamma_{\max}| ) \cdot \poly(\log \log n)) = O((\log |\Gamma_{\max}| ) \cdot \poly(\log \log n))$. 

Thus, the total cost to compute $I(y, z)$ (given the $I(y', z')$ for all pairs of partial strings of total length greater than $\ell$) is an $\textsf{AC}$ circuit of depth 
$O( (\log |\Gamma_{\max}|) \cdot \poly(\log \log n)).$ 

The result now follows.
\end{proof}

\section{Permutational Isomorphisms}
Two subgroups $G,H \leq S_n$ are considered \emph{permutationally isomorphic} if there is a bijection $\varphi \colon [n] \to [n]$ such that $G^{\varphi} := \{\varphi g \varphi^{-1} : g \in G\} = H$. Note that a permutational isomorphism induces an isomorphism of the abstract groups, namely the homomorphism that maps $g \in G$ to $\varphi g \varphi^{-1}$. The set of permutational isomorphisms $G \to H$ is denoted $\text{PISO}(G,H)$.

\begin{proposition}[cf. {\cite[Proposition~1.2.3]{CodenottiThesis}}] \label{prop:ListPermIsoTransitive}
Let $G, H \leq S_{n}$ be transitive permutation groups, and let $\phi \in \text{Iso}(G, H)$. We can list the (at most $n$) permutational isomorphisms $\pi \in \text{PISO}(G, H)$ corresponding to $\phi$ in $\textsf{NL}$.
\end{proposition}

Note that here the domain size is $n$, rather than $m \in O(\log n)$; that is, we are not working in the ``small domain'' setting in this section.

\begin{proof}
We first recall the procedure in the proof of \cite[Proposition~1.2.3]{CodenottiThesis}, and then show that it can be implemented in $\textsf{NL}$.
Let $S$ be the given set of generators for $G$. Then $\phi(S)$ generates $H$. We construct two directed, edge-colored graphs $X(G), X(H)$, over the vertex set $[n]$. For each $\sigma \in S$, we add to $X(G)$ an edge $(i, j)$ colored by $\sigma$ (we abuse notation by labeling the colors according to the generators) if $\sigma(i)=j$. Similarly, for each $\sigma \in S$, we add to $X(H)$ an edge $(i, j)$ colored by $\sigma$ if $\phi(\sigma)(i)=j$. We have from the proof of \cite[Proposition~1.2.3]{CodenottiThesis} that colored graph isomorphisms between $X(G)$ and $X(H)$ are precisely the permutational isomorphisms of $G$ and $H$.

In the proof of \cite[Proposition~1.2.3]{CodenottiThesis}, Codenotti decides isomorphism (constructively) between $X(G)$ and $X(H)$ in the following manner. Fix $x \in X(G), y \in X(H)$. Now apply BFS to $X(G)$ (resp. $X(H)$) starting from $x$ (resp. $y$). We assign to each vertex $v \in X(G)$ (resp. $v' \in X(H)$) the unique sequence of edge colors from $x$ (resp. $y$) to $v$ (resp. $v'$). If the labelings do not yield a bijection, we reject. Otherwise, we check whether this bijection is indeed an isomorphism between the two graphs.

We claim that this procedure can be implemented in $\textsf{NL}$. For fixed $x$, we consider each $v \in X(G)$. Now we can non-deterministically guess an $x \to v$ path, and label $v$ by the sequence of edge colors along this path. We proceed analogously for $X(H)$. Now we check in $\textsf{AC}^{0}$ whether (i) these labels induce a bijection, and (ii) if so, whether this bijection is a colored graph isomorphism of $X(G)$ and $X(H)$. Note that we need not check whether the labeling we guessed was witnessed by a breadth-first traversal of $X(G)$ and $X(H)$, for if $\phi$ is a permutational isomorphism, then any such path will satisfy the above check, and we only need a set of paths that span the graph to determine $\phi$ completely.

To enumerate all possible permutational isomorphisms corresponding to $\phi$, we need to try each choice for $x$ as the root. As there are at most $n$ permutational isomorphisms, we can enumerate all permutational isomorphisms in $\textsf{NL}$, as desired.
\end{proof}

\section{Permutational isomorphism of transitive permutation groups on small domains in parallel}

\subsection{Additional preliminaries: primitivity and structure trees} \label{sec:perm_prelim}
We refer to the textbook by Dixon \& Mortimer \cite{DixonMortimer} for many of the concepts around permutation groups that we briefly review here. A permutation group $G \leq S_m$ is \emph{transitive} if the permutation domain $[m]$ is a single $G$-orbit. An equivalence relation $\sim$ on $[m]$ is $G$-invariant if $x \sim y \Leftrightarrow x^g \sim y^g$ for all $x,y \in [m], g \in G$. The discrete equivalence---in which all elements are pairwise inequivalent---and indiscrete equivalence---in which all elements are equivalent---are considered trivial. The equivalence classes of any $G$-invariant equivalence relation are called \emph{blocks} of $G$. The singleton sets and the whole set $[m]$ are considered trivial blocks; any other blocks, if they exist, are non-trivial. $G$ is \emph{primitive} if it has no non-trivial blocks, or equivalently, has no non-trivial $G$-invariant equivalence relations. 

Equivalently, a block for $G$ is a subset $\Delta \subseteq [m]$ such that for every $g \in G$, $\Delta^g$ is either disjoint from $\Delta$ or equal to $\Delta$. A \emph{system of imprimitivity} for $G$ is a collection of blocks that partition $[m]$ (equivalently, the equivalence classes of a $G$-invariant equivalence relation). If $G$ is transitive, every system of imprimitivity arises as $\{\Delta^g : g \in G\}$ for some block $\Delta$. A block is \emph{minimal} if it is non-trivial and contains no proper subset that is also a non-trivial block. Two distinct minimal blocks can intersect in at most one point (standard exercise). An \emph{orbital} of $G$ is an orbit of $G$ on $[m] \times [m]$, the set of ordered pairs of elements of $[m]$. If we think of an orbital as a set of edges for a directed graph on vertex set $[m]$, then each connected component of this graph is a block for $G$. 

The following definitions tend to arise more in the setting of algorithms for permutation groups; see, e.\,g., \cite{Seress} for a textbook treatment.

\begin{definition}
An \textit{invariant tree} for a transitive group $G \leq \text{Sym}(\Omega)$ is a rooted tree whose set of leaves is $\Omega$ and to which the $G$-action extends as tree automorphisms. Note that such an extension is necessarily unique.
\end{definition}

\begin{definition}
An invariant tree for $G \leq \text{Sym}(\Omega)$ is a \textit{structure tree} if every internal node has at least $2$ children; and for every internal node $u$ of the tree, the action of $\text{Stab}_{G}(u)$ on the set of children of $u$ is primitive.
\end{definition}

An invariant tree corresponds to a hierarchy of block systems, and a structure tree corresponds to a maximal such hierarchy. Each node $v$ is labeled by the set $B(v)$ of leaves in the subtree rooted at $v$. The $k$th layer of a structure tree is the set of all nodes at distance $k$ from the root. The depth of the structure tree is $d$ if the leaves are on the $d$th layer. It follows from the definition of structure tree (and the transitivity of $G$) that all leaves must be on the same layer. 

Now note that every node corresponds to a block $B$ in a system of imprimitivity, and the blocks corresponding to nodes on the same layer are those in the $G$-orbit of $B$. This implies that if $u, v$ are nodes on the same layer, we have that $\text{Stab}_G(u) \cong \text{Stab}_G(v)$. In particular, a given layer of the tree is determined by any one block on that layer (by taking the orbit of that block).

\subsection{Parallel algorithms for small-domain structure trees and primitive groups}
Babai, Codenotti, and Qiao \cite[Lemma~2]{BCQ}=\cite[Lemma~1.2.11]{CodenottiThesis} showed that if $G \leq \text{Sym}(\Omega)$ is transitive of degree $m$, the number of structure trees is at most $m^{2\log m}$, and that they can be listed in time $m^{2\log(m) + O(1)}$. We show that this result can be done effectively in parallel.

\begin{lemma}[cf. {\cite[Lemma~2]{BCQ}=\cite[Lemma~1.2.11]{CodenottiThesis}}] \label{lem:EnumerateStructureTrees}
Let $G \leq \text{Sym}(\Omega)$ be a transitive permutation group of degree $m$. We can list all structure trees of $G$ using an $\textsf{SAC}$ circuit of depth $O(\log^2 m)$ and size $m^{O(\log m)}$.
\end{lemma}

\begin{proof}
We follow the strategy in the proof of \cite[Lemma~1.2.11]{CodenottiThesis}. Recall that every layer of a structure tree is determined by a single block on that layer. Thus, a structure tree of depth $\ell$ is determined by a nested sequence of blocks 
\[
B_{\ell} \subseteq B_{\ell-1} \subseteq \cdots \subseteq B_{0} = \Omega,
\] 
where $B_{\ell}$ consists of exactly one element of $\Omega$. As the $B_{i}$ are blocks of imprimitivity, $|B_{i}| \geq 2 \cdot |B_{i+1}|$, and so $\ell \leq \lceil \log_{2}(m) \rceil$. In order to obtain $B_{i}$ from $B_{i+1}$, we augment $B_{i}$ with a minimal block. As two minimal blocks intersect either trivially or in exactly one point, there are at most $\binom{m}{2}$ such minimal blocks and so at most $\binom{m}{2}$ such choices to obtain $B_{i}$ from $B_{i+1}$.

To construct all structure trees, we try every possible nested sequence of blocks. We now turn to finding all possible minimal blocks for a fixed layer. To do so, we consider the orbitals of $G$ (see Section~\ref{sec:perm_prelim}). For every orbital $\Sigma_j$ of $G$, construct a directed graph $X_j(\Omega, E_j)$, where $(x,y) \in E_j$ precisely if $(x,y) \in \Sigma_j$. We can obtain the minimal blocks containing $x \in \Omega$ in the following way. For each orbital digraph $X_j$, consider the connected component containing $x$. These form a set system of at most $m-1$ sets. The minimal sets in this system are precisely the minimal blocks containing $x$. We can compute the connected components of the $X_j$ as well as identify the minimal such components containing $x$ in $\textsf{NL} \subseteq \textsf{SAC}^{1}$ (where the complexity is measured with respect to $m$). Thus, for a fixed $B_{i+1}$, we can compute all possible $B_i$ in $\textsf{SAC}^{1}$. It follows that we can enumerate all structure trees with a $\textsf{SAC}$ circuit of depth $O(\log^2 m)$ and size $m^{O(\log m)}$.
\end{proof}

We will use the following result of Cameron. If $G \leq S_m$, then $G = A_m$ and $G = S_m$ are both referred to as ``giant'' subgroups; all other primitive permutation groups are called non-giant.

\begin{theorem}[{\cite[Theorem~6.1]{Cameron}}] \label{thm:Cameron}
Let $G \leq S_{m}$ be a primitive group. We have the following:
\begin{enumerate}[label=(\alph*)]
\item If $G$ is not a giant, 
and if $m$ is sufficiently large, then $|G| \leq m^{\sqrt{m}}$.
\item For every $\epsilon > 0$, there exists a $c(\epsilon)$ such that for every $m$, if $G \leq S_{m}$ and $|G| > 2^{m^{\epsilon}}$, then $G$ can be generated by $c(\epsilon)$ generators.
\end{enumerate}
\end{theorem}

When $m \in O(\log n)$, we have that $m^{\sqrt{m}} \in \poly(n)$. We use this to parallelize \textit{small} instances of \cite[Lemma~1.2.6]{CodenottiThesis}. We begin with the following helper lemma.

\begin{lemma} \label{lem:ListPrimitive}
Let $m \in O(\log n)$, and let $G \leq S_m$. We can list the elements of $G$ with an $\textsf{AC}$ circuit of depth $O((\log |G|) \cdot (\log \log n))$ and size $|G|\poly(n)$. If furthermore $G$ is primitive and not a giant, then its elements can be listed by a circuit of size $\poly(n)$ and depth $O(\sqrt{\log n} (\log \log n)^2)$.
\end{lemma}

In principle, a group $G \leq S_m$ could have order as large as $m!$, which is asymptotically $n^{\Theta(\log \log n)}$ if $m \in \Theta(\log n)$. However, we will only need to apply this lemma to $G \leq S_m$ that are primitive and not giant, so by Theorem~\ref{thm:Cameron}(1), such a group has size at most $m^{\sqrt{m}} \leq \poly(n)$.

\begin{proof}
We begin by computing $|G|$ in $\FOPLL$ (Lemma~\ref{PermutationGroupsNC}). Now let $S$ be the set of generators we are given for $G$. Wolf \cite{Wolf} showed that we can write down $\langle S \rangle$ in $O(\log |G|)$ rounds. We recall Wolf's procedure here to show how we will use it for the result. Let $T_{0} = S \cup \{1\}$. Now for $i \geq 0$, we construct $T_{i+1}$ by computing $xy$, for each $x, y \in T_{i}$; by a straightforward induction, $T_i$ consists of all group elements expressible as words in $S$ of length at most $2^i$. Our procedure terminates when $|T_{i}| = |G|$. 

We note that Wolf is considering quasigroups given by their multiplication tables, in which case we can multiply two elements in $\textsf{AC}^{0}$. However, we are dealing with permutation groups, where multiplication is $\textsf{DSPACE}(\log m)$-computable \cite{COOK1987385}. As $m \in O(\log n)$, we can multiply two permutations in $\mathsf{DSPACE}(O(\log \log n)) \subseteq \textsf{FOLL}$. Thus, given $T_{i}$, we can compute $T_{i+1}$ in $\FOLL$. The first part of the result follows.

To see the furthermore, now suppose that $G$ is primitive and not $S_m$ nor $A_m$. We have by Theorem~\ref{thm:Cameron} that $|G| \leq m^{\sqrt{m}} = (\log n)^{O(\sqrt{\log n})}$.  Note that $O(\log |G|) \subseteq O( \sqrt{\log n} \cdot (\log \log n))$. Thus, when $G$ is a non-giant primitive our total circuit depth becomes $O(\sqrt{\log n} \cdot (\log \log n)^2)$ and the size is $|G|\poly(n) \leq 2^{O(\sqrt{\log n} \log \log n)}\poly(n) \leq 2^{O(\log n)}\poly(n) = \poly(n)$.
\end{proof}

A \emph{permutational automorphism} of $G \leq S_m$ is a permutational isomorphism $G \to G$. We denote the set (actually, group) of permutational automorphisms of $G$ by $\text{PAut}(G)$.

\begin{lemma}[cf. {\cite[Lemma~1]{BCQ}, \cite[Lemma~1.2.6]{CodenottiThesis}}] \label{lem:PrimitivePIso}
Let $m \in O(\log n)$. Suppose $G \leq S_{m}$ is a primitive group other than $S_{m}$ or $A_{m}$. Then:
\begin{enumerate}[label=(\alph*)]
\item $|\text{PAut}(G)| \leq \exp(\tilde O(\sqrt{m})) \leq \poly(n)$, and
\item For every $H \leq S_{m}$, we can list $\text{PISO}(G, H)$ in $\mathsf{AC}^1$.
\end{enumerate}
\end{lemma}

\begin{proof}
We proceed as follows.
\begin{enumerate}[label=(\alph*)]
\item By \cite[Lemma~1]{BCQ},\cite[Lemma~1.2.6]{CodenottiThesis} and the fact that $m \in O(\log n)$, we have that $|\text{PAut}(G)| \leq \exp(\tilde{O}(\sqrt{m})) \in \poly(n)$.

\item \textbf{Step 1.} We first compute $|G|$ and $|H|$ in $\FOPLL$ (Lemma~\ref{PermutationGroupsNC}). If $|G| \neq |H|$, then $\text{PISO}(G, H) = \emptyset$. So suppose $|G| = |H|$. By Lemma~\ref{lem:ListPrimitive}, we can list $G, H$ using an $\textsf{AC}$ circuit of depth $O(\sqrt{\log n} \cdot (\log \log n)^2)$ and size $\poly(n)$. In particular, we now have $G, H$ given by their multiplication tables.

\textbf{Step 2.} Here we split into cases depending on $|G|$.  By Theorem~\ref{thm:Cameron}(b) with $\epsilon = 1/2$, there exists a constant $c$ such that if $|G| > 2^{m^{1/2}}$, then $G$ can be generated by $c$ elements. In this case, $|\text{Iso}(G, H)| \leq |G|^{c} \in\poly(n)$, and we may (using that $m \in O(\log n))$, list $\text{Iso}(G, H)$ in $\textsf{L}$ \cite{TangThesis} using the generator-enumeration strategy. 

Suppose instead that $|G| < 2^{m^{1/2}}$. In this case, we only know that $G$ might require at most $\log |G|$ many generators. Observe that $\log |G| < m^{1/2} \in O((\log n)^{1/2})$. Now as $|G| = |H|$, we have that:
\begin{align*}
|\text{Iso}(G,H)| &\leq |G|^{\log |G|} \\
&\leq (2^{m^{1/2}})^{m^{1/2}} \\
&= 2^{m} \\ 
&\in 2^{O( \log n)} 
\end{align*}

\noindent which is $\poly(n)$. Thus, we can again use generator-enumeration to list $\text{Iso}(G, H)$ in $\textsf{L}$ \cite{TangThesis}.

\textbf{Step 3.} Now for each $\phi \in \text{Iso}(G, H)$ in parallel, we use Proposition~\ref{prop:ListPermIsoTransitive} to enumerate the at most $n$ permutational isomorphisms between $G$ and $H$ that correspond to $\phi$. This step is $\textsf{NSPACE}(\log m)$-computable. As $m \in O(\log n)$, this yields a bound of $\textsf{FOLL}$ (by considering each $\phi \in \text{Iso}(G, H)$ in parallel) for listing all permutational isomorphisms.

\textbf{Calculating the depth.} Since each of the three steps can be done sequentially, their depths add, resulting in a total depth of $O(\sqrt{\log n} (\log \log n)^2 + \log n + \log \log n) \leq O(\log n)$. As each step was also polynomial size, the whole circuit is polynomial size, thus giving us $\mathsf{AC}^1$ as claimed. \qedhere
\end{enumerate}
\end{proof}

\subsection{Parallel algorithms for small-domain transitive permutation groups}

In this section, we adapt \cite[Theorem~3]{BCQ}=\cite[Theorem~5.2.1]{CodenottiThesis} to the setting where the degree of the permutation group $G$ is \textit{small} (size $O(\log n)$) and $|G| \leq n$. This is precisely the setting we will need for $\textsf{NC}$ isomorphism testing of Fitting-free groups.

\begin{theorem}[cf. {\cite[Theorem~3]{BCQ}=\cite[Theorem~5.2.1]{CodenottiThesis}}] \label{thm:5.2.1} \label{thm:TransitivePermIso}
Let $m \in O(\log n)$, and let $G, H \leq S_{m}$ be transitive. Then:  
\begin{enumerate}[label=(\alph*)]
\item $|\text{PAut}(G)| \leq |G| \cdot \poly(n)$, and 
\item $\text{PISO}(G, H)$ can be listed with an $\textsf{AC}$ circuit of depth $O( (\log n) (\log \log n)^2 )$ and size $\poly(n)$.
\end{enumerate}
\end{theorem}

To prove Theorem~\ref{thm:TransitivePermIso}, we follow the strategy in \cite[Section~4.2]{BCQ} and \cite[Section~5.2]{CodenottiThesis}. Let $T$ be a structure tree of $G$, and let $U$ be a structure tree of $H$. We say that $T$ and $U$ are \textit{compatible} if their depth is the same; and for every $\ell$, the primitive groups arising on level $\ell$ in $G$ and $H$ (as actions of the stabilizers of a node on level $\ell$ on the children of that node) are permutationally isomorphic. For non-giant primitive groups, we use Lemma~\ref{lem:PrimitivePIso} to bound the number of permutational isomorphisms and list them in parallel. For giant primitive groups, we will need more careful structural analysis.

At the inductive step, we will consider the following situation, which corresponds to one layer of the structure tree. Let $G \leq \text{Sym}(\Omega), H \leq \text{Sym}(\Delta)$ with $m := |\Omega| = |\Delta| \in O(\log n)$. Let $\{ \Omega_1, \ldots, \Omega_{\ell}\}$ be a system of imprimitivity for $G$, and let $K$ be the kernel of the action on the set of blocks. Let $G^{*}$ be the action on the set of blocks. Similarly, let $\{ \Delta_1, \ldots, \Delta_{\ell} \}$ be a system of imprimitivity for $H$, let $L$ be the kernel of the $H$-action on the set of blocks, and let $L^{*}$ be the $H$-action on the set of blocks. Note that $K$ is the kernel of the restriction homomorphism $G \to G^{*}$. Throughout this section, we will use $k$ to be the size of the blocks $\Omega_{i}$ and $\Delta_{j}$. That is, for all $i, j$, we have $k = |\Omega_{i}| = |\Delta_{j}|$.

We say that $\phi \in \text{PISO}(G, H)$ \textit{extends} $\pi \in \text{PISO}(G^{*}, H^{*})$ if $\phi$ acts as $\pi$ on the set of blocks; that is, if $\phi(\Omega_i) = \Delta_{i^{\pi}}$ for all $i \in [\ell]$. Let:
\[
\text{PISO}(G, H; \pi) = \{ \phi \in \text{PISO}(G, H) \mid \phi \text{ extends } \pi \}.
\]

Let $G(i) := G_{\{\Omega_i\}}^{\Omega_{i}}$ and $K(i) := K^{\Omega_{i}}$ be the restrictions of the groups $G$ and $K$ respectively to the block $\Omega_i$. Define $H(i)$ and $L(i)$ analogously.

For the case when $K(i)$ is a giant, Babai, Codenotti, and Qiao \cite[Section~4.2]{BCQ} used \cite[Lemma~5.2.2]{CodenottiThesis} to bound the number of extensions and show that they can be listed efficiently. We adapt this lemma as follows.

To state the lemma and for its proof, we need a few additional preliminaries. Given groups $G_1, \dotsc, G_\ell$, a \emph{subdirect product} is a subgroup $G \leq G_1 \times \dotsb \times G_\ell$ such that the projection of $G$ onto the $i$-th factor $\pi_i \colon G \to G_i$ is surjective for all $i$. If $P_i \leq \text{Sym}(\Omega_i)$ are permutation groups that are each permutationally isomorphic to some group $T \leq \text{Sym}(\Omega)$, a \emph{permutation diagonal} \cite[Section~4.1]{BCQ} \cite[Def.~1.2.13]{CodenottiThesis} is a list of permutational isomorphisms $\varphi_i \in \text{PISO}(T, P_i)$. Then $f_{ij} := \varphi_i^{-1} \varphi_j \in \text{PISO}(P_i, P_j)$ are referred to as the bijections defining this diagonal.

\begin{lemma}[cf. {\cite[Section~4.2]{BCQ},\cite[Lemma~5.2.2]{CodenottiThesis}}] \label{lem:5.2.2}
Let $m = |\Omega|=|\Delta| \in O(\log n)$, and  $\{ \Omega_{1}, \ldots, \Omega_{\ell} \}$ and $\{ \Delta_{1}, \ldots, \Delta_{\ell}\}$, and $k=|\Omega_i|=|\Delta_j|$ as defined above. Suppose that $k \geq 5$, with $k \neq 6$. Suppose that $K \leq \prod_{i=1}^{\ell} \text{Alt}(\Omega_i)$ and $L \leq \prod_{i=1}^{\ell} \text{Alt}(\Delta_i)$ are subdirect products. Given $\pi \in S_{\ell}$, we can list all $\leq 2^\ell \cdot |K|$ permutational isomorphisms $\phi \in \text{PISO}(K, L; \pi)$ by uniform $\cc{AC}$ circuits of depth $O(\log |K| \log \log n + (\log \log n)^{O(1)})$ and size $|K|\poly(n)$.
\end{lemma}

\begin{proof}
As mentioned in \cite[Section~4.2]{BCQ}, and proved in full in \cite[Lemma~5.2.2]{CodenottiThesis}, we have that $|\text{PISO}(K, L; \pi)| \leq 2^\ell \cdot |K|$. We start by computing, in $\FOPLL$ (Lemma~\ref{PermutationGroupsNC}), $|K|$ and $|L|$. For if $|K| \neq |L|$, then $\text{PISO}(K,L) = \emptyset$ and we are done. So suppose $|K| = |L|$. We can (by Lemma~\ref{lem:ListPrimitive}) write down $K$ and $L$ with an $\textsf{AC}$ circuit of depth $O((\log |K|)(\log \log n))$ and size $|K|\poly(n)$.

Now, by \cite[Fact~1.2.15]{CodenottiThesis}, $K$ and $L$ are direct products of permutation diagonals. As abstract groups, we have that $K$ and $L$ are isomorphic to direct products of alternating groups $A_n$ with $n \geq 5, n \neq 6$. Since such alternating groups are simple, the direct product decomposition of $K$ (resp., $L$) is unique. For such groups, in $\textsf{DSPACE}(\log |K|)$ we can compute the direct product decomposition \cite[Lem.~6.11 in the arXiv version]{GrochowLevetWL}. Each abstract direct factor of $K$ corresponds to a permutation diagonal direct factor of $K$ as a permutation group. Thus in $\textsf{DSPACE}(\log|K|)$ we have decomposed $K$ (resp., $L$) into a direct product of permutation diagonals. In $\textsf{DSPACE}(\log n)$, we may also check that the decomposition of $K$ and $L$ have the same numerical parameters: for each $i$, the number of direct factors of $K$ with $i$ orbits is the same as that in $L$, for otherwise $\text{PISO}(K,L) = \emptyset$ and we are done. $\cc{DSPACE}(O(\log|K|))$ can be simulated by $\cc{AC}$ circuits of depth $O(\log |K|)$ and size $O(|K|)$.

Let $\sim_{K}$ be the equivalence relation on $[\ell]$ corresponding to the partition into permutation diagonals of $K$, and define $\sim_{L}$ analogously for $L$. For $i \sim_{K} j$, let $f_{ij} : \Omega_i \to \Omega_j$ defining the permutational diagonal; and for $i \sim_{L} j$, define $g_{ij} : \Delta_i \to \Delta_j$ analogously. Now let $R \subseteq [\ell]$ be a set of representatives of the $\sim_K$ equivalence classes. Any list of bijections $(\phi_{i} : \Omega_{i} \to \Delta_{\pi(i)})_{i \in R}$ uniquely defines the remaining $\phi_{j}$, and therefore all of $\phi$, using the formula
\[
\phi_{j} = g_{\pi(i), \pi(j)} \phi_{i} f_{i,j}^{-1}
\]
 (see the proof of \cite[Lemma~5.2.2]{CodenottiThesis}).
Once we have decomposed $K$ (resp., $L$) into a direct product of permutational diagonals, we can enumerate all possible lists $(\phi_i)_{i \in R}$ in parallel with an $\textsf{AC}$ circuit of constant depth. As there are $2^{\ell} \cdot |K|$ such lists \cite[Lemma~5.2.2]{CodenottiThesis}, and each list takes $|R| \cdot \log(k!) \leq \ell k \log k \leq m \log k$ bits to write down, our circuit has size at most $m \log(k) 2^\ell |K| \leq m \log(m) 2^m |K|$. As $m \in O(\log n)$, the size of our circuit is $|K|\poly(n)$.

Now given such a list $(\phi_i)_{i \in R}$, we can construct the remaining $\phi_{j}$ in $\textsf{L}$ using the formula for $\phi_j$ above. Since $\phi \colon \Omega \to \Delta$ is just the concatenation of the $\phi_j$, we can write down $\phi$ using additional constant depth. By Lemma~\ref{PermutationGroupsNC}, we can test whether $\phi$ is a permutational isomorphism in $\FOPLL$. Since the depths of the various parts of this algorithm add, we get total depth $O(\log |K| \log \log n + \log |K| + c + (\log \log n)^c) \leq O(\log |K| \log \log n + \poly(\log \log n))$. The result now follows.
\end{proof}

We now recall the following results, which we will use when $G(i)$ is a giant, but $K(i)$ is trivial.

\begin{observation}[{\cite[Observation~5.2.3]{CodenottiThesis}}] \label{obs:5.2.3}
If $G$ acts transitively on the blocks $\Omega_{1}, \ldots, \Omega_{\ell}$, and there exists some $i$ such that $K(i) = \{1\}$, then $K = \{1\}$. 
\end{observation}

\begin{lemma}[{\cite[Lemma~5.2.4]{CodenottiThesis}}] \label{lem:5.2.4}
If the system of imprimitivity $\{\Omega_1, \ldots, \Omega_{\ell}\}$ is maximal and two distinct points $x \neq y$ belong to the same block $\Omega(i)$, and $G(i) \neq \{1\}$, then $\text{Stab}_{G}(x) \neq \text{Stab}_{G}(y)$. 
\end{lemma}

\begin{lemma}[cf. {\cite[Lemma~5.2.5]{CodenottiThesis}}] \label{lem:5.2.5}
Suppose there exists an $i$ s.t. $G(i) \neq \{1\}$ and $K(i) = \{1\}$. Then there exists at most one extension for a given $\pi \in \text{PISO}(G^{*}, H^{*})$. Furthermore, if such an extension exists, we can find it in $\FOPLL$.
\end{lemma}

\begin{proof}
Babai, Codenotti, and Qiao (see the proof of \cite[Lemma~5.2.5]{CodenottiThesis}) established that there exists at most one extension for a given $\pi \in \text{PISO}(G^{*}, H^{*})$. In order to observe the $\FOPLL$ bound, we recall their argument. By Observation~\ref{obs:5.2.3}, we have that $K = \{1\}$. Thus, $G \cong G^{*}$. If $L \neq \{1\}$, no extension of $\pi$ exists. So without loss of generality, suppose that $L = \{1\}$. Thus, $L \cong L^{*}$. Then we have that for every $x \in \Omega_{i}$ and every $\phi \in \text{PISO}(G, H; \pi)$, there exists $y \in \Delta_{\pi(i)}$ such that $\phi(\text{Stab}_{G}(x)) = \text{Stab}_{H}(y)$. By Lemma~\ref{lem:5.2.4}, there in fact exists a unique such $y$. Thus, given $\pi$, we have a unique extension $\phi$.

We now turn to constructing $\phi$. Let $\sigma : G \to H$ be the group isomoprhism induced by $\pi$. We consider the $\Omega_i$ in parallel. For each such $i \in [m]$, check in parallel all $y \in \Delta_{\pi(i)}$ to see if $\sigma(\text{Stab}_{x}(G)) = \text{Stab}_{H}(y)$. By Lemma~\ref{PermutationGroupsNC}, we can compute $\text{Stab}_{G}(x)$ and $\text{Stab}_{H}(y)$ in $\FOPLL$. Then we can check if $\sigma(\text{Stab}_{G}(x)) = \text{Stab}_{H}(y)$ in $\FOPLL$, using a membership test (Lemma~\ref{PermutationGroupsNC}). If such a $y$ exists, we set $\phi(x) = y$. The result now follows.
\end{proof}

\begin{lemma}[cf. {\cite[Section~4.2]{BCQ}, \cite[Lemma~5.2.6]{CodenottiThesis}}] \label{lem:5.2.6}
Let $G \leq \text{Sym}(\Omega), H \leq \text{Sym}(\Delta)$ be transitive permutation groups of degree $O(\log n)$. Let $\Omega_{1}, \ldots, \Omega_{\ell}$ and $\Delta_{1}, \ldots, \Delta_{\ell}$ be maximal systems of imprimitivity for $G$ and $H$, respectively. Let $G^{*}, H^{*}$ be the actions of $G$ and $H$, respectively, on the blocks. For any $\pi \in \text{PISO}(G^{*}, H^{*})$, we can list $\text{PISO}(G, H; \pi)$ by $\cc{AC}$ circuits of depth $O(\log|G| \log \log n + (\log \log n)^{O(1)})$ and size $|G|\poly(n)$.
\end{lemma}

\begin{proof}
We follow the strategy of \cite[Lemma~5.2.6]{CodenottiThesis}. Suppose that for all $i, j$, $G(i), H(j)$ (notation as at the beginning of this section) are permutationally isomorphic to some primitive $T \leq S_{k}$. Otherwise, there is no $\phi \in \text{PISO}(G, H; \pi)$. We break into cases:
\begin{itemize}
\item \textbf{Case 1:} Suppose that $T$ is not a giant. We start by considering all possible lists of permutational isomorphisms $(\sigma_{i} \in \text{PISO}(G(i), H(\pi(i))))_{i \in [\ell]}$. Given such a list and $\pi$, there is a unique bijection $\phi : \Omega \to \Delta$ given by $\phi(\Omega_{i}) = \Delta_{\pi(i)}$, such that $\phi|_{\Omega_{i}} = \sigma_{i}$. By Lemma~\ref{PermutationGroupsNC}, we can check whether $\phi \in \text{PISO}(G, H)$ in $\FOPLL$.

Babai, Codenotti, and Qiao \cite[Lemma~1.2.6(a)]{CodenottiThesis} established that there exists a constant $c$ such that $|\text{PISO}(G(i), H(j))| \leq c^{k}$. Thus, the number of extensions is bounded by $(c^{k})^{\ell} = c^{O(\log n)} = \poly(n)$. By Lemma~\ref{lem:PrimitivePIso} we can list the at most $c^{\tilde O(\sqrt{k})} \leq c^{O(k)}$ elements of $\text{PISO}(G(i), H(\pi(i)))$ in $\cc{AC}^1$, and then we can list $\ell$-tuples of such elements with an additional $O(1)$ depth and total size at most $(c^{O(k)})^\ell \leq \poly(n)$.

\item \textbf{Case 2:} Suppose instead that $T$ is a giant. As $K \trianglelefteq G$, we have that $K(i) \trianglelefteq G(i)$. Similarly, $L(i) \trianglelefteq H(i)$. As $G(i)$ is permutationally isomorphic to $A_{k}$ or $S_{k}$, we have that $K(i)$ is either $S_{k}$, $A_{k}$, or $\{1\}$. If there exists some $i$ such that $K(i) = \{1\}$, then by \cite[Lemma~5.2.5]{CodenottiThesis}, there is at most one extension $\phi$ of $\pi$. By Lemma~\ref{lem:5.2.5}, we can construct $\phi$ in $\FOPLL$.

Suppose instead that $K(i) \neq \{1\}$. It was previously established in the proof of \cite[Lemma~5.2.6]{CodenottiThesis} that if $\phi \in \text{PISO}(G, H; \pi)$, then $\phi \in \text{PISO}([K,K], [L,L]; \pi)$ and $[K,K], [L,L]$ are subdirect products of $A_{k}^{m}$. For $k \geq 5$, $k \neq 6$, we have by \cite[Lemma~5.2.2]{CodenottiThesis} that there are at most $2^{O(\log n)} \cdot |[K,K]| \leq 2^{O(\log n)} \cdot |K|$ extensions.  By Lemma~\ref{lem:5.2.2}, we can list $\text{PISO}([K,K], [L,L]; \pi)$ using $\cc{AC}$ circuits of depth $O(\log|K| \log \log n + (\log \log n)^{O(1)})$ and size $|K|\poly(n)$. Then by Lemma~\ref{PermutationGroupsNC}, we can check in $\FOPLL$ whether each $\phi \in \text{PISO}([K,K], [L,L]; \pi)$ belongs to $\text{PISO}(K,L; \pi)$.
\end{itemize}

The result now follows.
\end{proof}

Following Babai, Codenotti, and Qiao, we will adapt \cite[Theorem~5.2.7]{CodenottiThesis}, which is a slightly stronger version of Theorem~\ref{thm:5.2.1} and will serve as the key to proving Theorem~\ref{thm:5.2.1}.

\begin{theorem}[cf. {\cite[Theorem~5.2.7]{CodenottiThesis}}] \label{thm:5.2.7}
Let $m \in O(\log n)$. Let $G, H$ be transitive permutation groups of degree $m$, given by their multiplication tables. Fix structure trees $T_{G}, T_{H}$ for $G, H$, respectively. The number of permutational isomorphisms respecting $T_{G}$ and $T_{H}$ is at most $|G| \cdot c^{2m}$. All permutational isomorphisms of $G$ and $H$ that respect $T_{G}$ and $T_{H}$ can be listed with an $\textsf{AC}$ circuit of depth $O((\log n) (\log \log n)^2)$ and size $\poly(n)$.
\end{theorem}

\begin{proof}
We proceed by induction on the depth $d$ of the structure trees. From the proof of Lemma~\ref{lem:EnumerateStructureTrees}, we have that the depth of a structure tree is at most $\lceil \log_{2}(m) \rceil \in O(\log \log n)$. Furthermore, Babai, Codenotti, and Qiao established that the number of permutational isomorphisms respecting $T_{G}$ and $T_{H}$ is $|G| \cdot c^{2m}$ \cite[Theorem~5.2.7]{CodenottiThesis}. 

The base case is when $d = 1$. Here, $G, H$ are primitive. If $G, H$ are not giants, then by Lemma~\ref{lem:PrimitivePIso}, we can list all permutational isomorphisms respecting $T_{G}$ and $T_{H}$ in $\cc{AC}^1$. If both groups are $A_{m}$ or $S_{m}$, then $\text{PISO}(G, H) = S_{m}$, and $m! \leq 2 \cdot |G|$. In this case, we may use the fact that $S_{m}$ is $2$-generated and Lemma~\ref{lem:ListPrimitive} to write down $\text{PISO}(G, H)$ using $\cc{AC}$ circuits of size $\poly(|G|)$ and depth $O(\log |G| \log \log n) = O(m \log m \log \log n) = O((\log n) (\log \log n)^2)$. If none of the above cases hold, then $\text{PISO}(G,H) = \emptyset$. 

Now consider structure trees of depth $d$. The last layer of a structure tree is a maximal system of imprimitivity. Let $\Omega := \{ \Omega_1, \ldots, \Omega_{\ell} \}$ be this system of imprimitivity for $G$, and let $\Delta := \{ \Delta_1, \ldots, \Delta_{\ell}\}$ be the corresponding system of imprimitivity for $H$. Let $G^{*}$ be the action of $G$ on the blocks of $\Omega$, and let $H^{*}$ be the action of $H$ on the blocks of $\Delta$. By \cite[Lemma~5.2.6]{CodenottiThesis}, $|\text{PISO}(G^{*}, H^{*})| \leq |G^{*}| \cdot c^{2m}$. By the inductive hypothesis, we can list them using $\cc{AC}$ circuits of depth $O( (\log n)((d-1)\log \log n + (\log \log n)^2) )$ and size $\poly(|G|) \leq \poly(n)$.

Now by \cite[Lemma~5.2.6]{CodenottiThesis}, we have that for each $\pi \in \text{PISO}(G^{*}, H^{*})$, $|\text{PISO}(G, H; \pi)| \leq |G^{*}| \cdot c^{2m} \in \poly(n)$. By Lemma~\ref{lem:5.2.6}, we can list each $\text{PISO}(G, H; \pi)$ in parallel using $\cc{AC}$ circuits of depth $O(\log|G| \log \log n + (\log \log n)^{O(1)})$ and size $|G|\poly(n)$. Babai, Codenotti, and Qiao previously established that $|\text{PISO}(G, H)| \leq |G| \cdot c^{2m} \in \poly(n)$. Thus, given $\text{PISO}(G^{*}, H^{*})$, we can list $\text{PISO}(G, H)$ using $\cc{AC}$ circuits of size $O(|G|^2\poly(n)) \leq \poly(n)$ and depth $O(\log |G| \log \log n + (\log \log n)^{O(1)} \leq O(\log n \log \log n)$.

As $d \leq \log_{2} m \in O(\log \log n)$, the above iterates at most $d$ times, and we can list $\text{PISO}(G, H)$ with an $\textsf{AC}$ circuit of depth $O( (\log n) (\log \log n)^2)$ and size $\poly(n)$.
\end{proof}

We now prove Theorem~\ref{thm:5.2.1}.

\begin{proof}[Proof of Theorem~\ref{thm:5.2.1}]
Babai, Codenotti, and Qiao \cite{BCQ} previously established that $|\text{PAut}(G)| \leq |G| \cdot c^{m}$. We now turn to proving (b). We construct a structure tree for $G$, and try all possible structure trees for $H$. As $m \in O(\log n)$, we have by Lemma~\ref{lem:EnumerateStructureTrees} that we can enumerate all $m^{O(\log m)} \subseteq (\log n)^{O(\log \log n)}$ structure trees in $\FOPLL$. The result now follows by Theorem~\ref{thm:5.2.7}.
\end{proof}

\section{Fitting-free group isomorphism in parallel} \label{sec:FF}

We now come to the isomorphism problem for Fitting-free groups. Precisely, we will establish the following.

\begin{theorem} \label{thm:FittingFree}
Let $G$ be a Fitting-free group of order $n$, and let $H$ be arbitrary. We can decide isomorphism between $G$ and $H$ with an $\textsf{AC}$ circuit of depth $O((\log^2 n) \cdot \poly(\log \log n))$ and size $\poly(n)$. 
\end{theorem}

\subsection{Preliminaries}

We begin by recalling some preliminaries. In \cite{BCGQ, BCQ}, the authors consider isomorphisms of $\Soc(G)$ respecting \textit{diagonal} products.

\begin{definition}[{\cite[Definition~2.2]{BCGQ}}]
Let $V_1, \ldots, V_r$ be isomorphic groups, where each $V_i \cong T$. A \textit{diagonal} of $(V_1, \ldots, V_r)$ is an embedding $\phi : T \to \prod_{i=1}^{r} V_i$ such that $\text{Im}(\phi)$ is a subdirect product of the $V_i$. More generally, suppose we have a system of groups $(V_{1,1}, \ldots, V_{1,k_{1}}), \ldots, (V_{r, 1}, \ldots, V_{r,k_{r}})$, where for every $i \leq r$ and every $j \leq k_{i}$, $V_{i,j} \cong T_{i}$. A \textit{diagonal product} of the system is an embedding:
\[
\phi_{1} \times \cdots \times \phi_{r} : T_1 \times \cdots \times T_r \to \prod_{j=1}^{k_{1}} V_{1,j} \times \cdots \times \prod_{j=1}^{k_{r}} V_{r,j},
\]
where each $\phi_{i}$ is a diagonal of $(V_{i, j})_{j=1}^{k_{i}}$. The \textit{standard diagonal} of $T^k$ is the map $\Delta: t \mapsto (t, \ldots, t)$. Similarly, the \textit{standard diagonal product} of $\prod_{i=1}^{r} T_{i}^{k_{i}}$ is the map $\Delta = \Delta_{1} \times \cdots \times \Delta_{r}$, where each $\Delta_{i}$ is the standard diagonal of $T_{i}^{k_{i}}$.
\end{definition}

We will be particularly interested in diagonals that respect the (unique) decomposition of $\Soc(G)$ into a direct product of non-Abelian simple groups.

\begin{definition}[{\cite[Definition~2.3]{BCGQ}}]
Let $X, Y$ be groups, where $X = \prod_{i=1}^{r} \prod_{j=1}^{k_{i}} V_{i,j}$, and $Y = \prod_{i=1}^{r} \prod_{j=1}^{k_{i}} U_{i,j}$, where for all $i \in [r], j \in [k_{i}]$, $U_{i,j} \cong V_{i,j} \cong T_{i}$. We say that an isomorphism $\chi : G \to H$ \textit{respects} the decompositions $\mathcal{V} = (V_{i,j})$ and $\mathcal{U} = (U_{i,j})$ if for every $i \in [r]$ and every $j \in [k_{i}]$, there exists $j' \in [k_{i}]$ such that $\chi(V_{i,j}) = U_{i,j'}$. We denote the set of isomorphisms that respect the decompositions $\mathcal{V}$ and $\mathcal{U}$ by $\text{ISOp}((X, \mathcal{V}), (Y, \mathcal{U}))$, where $p$ stands for product decomposition. If the decompositions are understood from context, we will write $\text{ISOp}(X,Y)$.
\end{definition}

\begin{definition}[{\cite[Definition~2.4]{BCGQ}}]
Let $X, Y$ be two groups with direct product decompositions $X = \prod_{i=1}^{r} \prod_{j=1}^{k_{i}} V_{i,j}$ and $Y = \prod_{i=1}^{r} \prod_{j=1}^{k_{i}} U_{i,j}$, where for all $i,j$, $U_{i,j} \cong V_{i,j} \cong G_{i}$. Let $\varphi$ be a diagonal product of $\mathcal{V} = (V_{i,j})$, and let $\psi$ be a diagonal product of $\mathcal{U} = (U_{i,j})$. We say that an isomorphism $\chi \in \text{ISOp}((X, \mathcal{V}), (Y, \mathcal{U}))$ \textit{respects} the diagonal products $\varphi, \psi$ if $\varphi \circ \chi = \psi$. We denote the set of diagonal product respecting automorphisms by $\text{ISOd}((X, \mathcal{V}), (Y, \mathcal{U}); \varphi, \psi)$.
\end{definition}

\begin{lemma}[{\cite[Lemma~2.2]{BCGQ}}]
Let $X, Y$ be two groups with direct product decompositions $X = \prod_{i=1}^{r} \prod_{j=1}^{k_{i}} V_{i,j}$ and $Y = \prod_{i=1}^{r} \prod_{j=1}^{k_{i}} U_{i,j}$, where for all $i,j$, $U_{i,j} \cong V_{i,j} \cong G_{i}$. Let $\varphi$ be a diagonal product of $\mathcal{V} = (V_{i,j})$, and let $\mathcal{D}$ be the set of all diagonal products of $\mathcal{U}$. Then:
\[
\text{ISOp}((X, \mathcal{V}), (Y, \mathcal{U})) = \bigcup_{\psi \in \mathcal{D}} \text{ISOd}((X, \mathcal{V}), (Y, \mathcal{U}); \varphi, \psi).
\]
\end{lemma}

For Fitting-free groups, we can reduce to considering diagonal-respecting isomorphisms of the socle. Before we can make this precise, we must first recall the following lemma. 

\begin{lemma}[{\cite[Lemma 3.1]{BCGQ}, cf. \cite[\S 3]{CH03}}] \label{CharacterizeSemisimple}
Let $G$ and $H$ be groups, with $R \triangleleft G$ and $S \triangleleft H$ groups with trivial centralizers. Let $\alpha : G \to G^{*} \leq \Aut(R)$ and $\beta : H \to H^{*} \leq \Aut(S)$ be faithful permutation representations of $G$ and $H$ via the conjugation action on $R$ and $S$, respectively. Let $f : R \to S$ be an isomorphism. Then $f$ extends to an isomorphism $\hat{f} : G \to H$ if and only if $f$ is a permutational isomorphism between $G^{*}$ and $H^{*}$; and if so, $\hat{f} = \alpha f^{*} \beta^{-1}$, where $f^{*} :  G^{*} \to H^{*}$ is the isomorphism induced by $f$.
\end{lemma}

Let $G, H$ be Fitting-free groups. Let $R = \Soc(G)$ and $S = \Soc(H)$. For diagonal products $\phi, \psi$ of $R, S$ respectively, we denote the set of isomorphisms respecting the diagonal products of the socle as:
\[
\text{ISOds}(G, H; \phi, \psi) = \{ \chi \in \text{ISO}(G, H) : \chi|_{\Soc(G)} \in \text{ISOd}(R, S; \varphi, \psi) \}.
\]
In particular, we have the following:

\begin{corollary}[{\cite[Corollary~4.1]{BCGQ}}] \label{cor:BCGQ4.1}
Let $G, H$ be Fitting-free, $\phi$ a diagonal product of $\Soc(G)$, and $\mathcal{D}$ the set of diagonal products of $\Soc(H)$. Then:
\[
\text{ISO}(G, H) = \bigcup_{\psi \in \mathcal{D}} \text{ISOds}(G, H; \varphi, \psi).
\]
\end{corollary}

\begin{lemma}[{\cite[Lemma~4.1]{BCGQ}}] \label{lem:4.1}
Let $H, \mathcal{D}$ be defined as in Corollary~\ref{cor:BCGQ4.1}. We have that $|\mathcal{D}| \leq |H|^{2}$ .
\end{lemma}

Babai, Codenotti, and Qiao established the following lemma, relating permutational isomorphisms of $G/\text{PKer}(G)$ and $H/\text{PKer}(H)$ to isomorphisms of $G$ and $H$.

\begin{lemma}[{\cite[Lemma~6.3.10]{CodenottiThesis}}] \label{lem:6.3.10}
Let $G, H$ be Fitting-free groups given by their Cayley tables. Suppose that $\Soc(G) \cong \Soc(H)$. Let $P := G/\text{PKer}(G)$ and $Q := H/\text{PKer}(H)$. Let $\varphi, \psi$ be diagonal products of $G, H$ respectively. Then:
\begin{enumerate}[label=(\alph*)]
\item  Every isomorphism $\chi \in \text{ISOds}(G, H; \varphi, \psi)$ is determined by the permutational isomorphism it induces between $P$ and $Q$; and
\item Given a permutational isomorphism $f \in \text{PISO}(P, Q)$, we can check whether it arises as the action of some $\chi \in \text{ISOds}(G, H; \varphi, \psi)$, and if so, find that unique $\chi$, in polynomial time.
\end{enumerate}
\end{lemma}

We now parallelize Lemma~\ref{lem:6.3.10}(b).

\begin{lemma} \label{lem:6.3.10Parallel}
Let $G, H$ be Fitting-free groups given by their Cayley tables. Suppose that $\Soc(G) \cong \Soc(H)$. Let $P := G/\text{PKer}(G)$ and $Q := H/\text{PKer}(H)$. Let $\varphi, \psi$ be diagonal products of $G, H$ respectively. Given a permutational isomorphism $f \in \text{PISO}(P, Q)$, we can check whether it arises as the action of some $\chi \in \text{ISOds}(G, H; \varphi, \psi)$, and if so, find that unique $\chi$, in \FOPLL.
\end{lemma}

\begin{proof}
We write:
\begin{align*}
&\Soc(G) = \prod_{i=1}^{r} \prod_{j=1}^{k_{i}} V_{i,j}, \text{ and } \,
\Soc(H) = \prod_{i=1}^{r} \prod_{j=1}^{k_{i}} U_{i,j},
\end{align*}

\noindent where for all $i, j$, $V_{i,j} \cong U_{i,j} \cong T_{i}$, where $T_{i}$ is non-Abelian simple. Let $\varphi = \prod_{i=1}^{r} \varphi_{i}$, where $\varphi_{i} : T_{i} \xhookrightarrow{} \prod_{j=1}^{k_{i}} V_{i,j}$ is a diagonal. Similarly define $(\psi_{i})_{i=1}^{r}$ for $H$.  Now for all $i, j$, let $\varphi_{ij} : T_{i} \to V_{i,j}$ be the projection of $\varphi_{i}$ to the $j$th coordinate of $\prod_{h=1}^{k_{i}} V_{i,h}$. Similarly, let $\psi_{ij} : T_{i} \to U_{i,j}$ be the projection of $\psi_{i}$ to the $j$th coordinate of $\prod_{h=1}^{k_{i}} U_{i,h}$. In particular, observe that $\varphi_{ij}$ is an isomorphism between $T_{i}$ and $V_{i,j}$, and similarly $\psi_{ij}$ is an isomorphism between $T_{i}$ and $U_{i,j}$.

We construct $\chi_{f} : \Soc(G) \to \Soc(H)$ as follows. For every $i, j$, let $\ell$ such that $f(V_{i,j}) = U_{i,\ell}$. Set $\chi_{f}|_{V_{i,j}} = \psi_{ij}^{-1} \circ \psi_{i\ell}$. Babai, Codenotti, and Qiao established that $\chi_{f} \in \text{ISOd}(\Soc(G), \Soc(H); \varphi, \psi)$. By Lemma~\ref{PermutationGroupsNC}, we can check in $\FOPLL$ whether $\chi_{f}$ extends to an isomorphism of $G$ and $H$.
\end{proof}

\subsection{Fitting-free groups with unique minimal normal subgroup}

We now turn to proving Theorem~\ref{thm:FittingFree}. We first show that if $H$ is not Fitting-free, we can easily distinguish $G$ from $H$.

\begin{lemma} \label{lem:IsFittingFree}
We can decide if a group $G$ is Fitting-free in $\textsf{AC}^{0}$.
\end{lemma}

\begin{proof}
$G$ is Fitting-free precisely if it has no Abelian normal subgroups. For each group element $x \in G$, in parallel, we write down generators $\{ gxg^{-1} : g \in G \}$ for $\ncl_{G}(x)$. This step is $\textsf{AC}^{0}$-computable. Now $\ncl_{G}(x)$ is Abelian if and only if $g_{1}xg_{1}^{-1}, g_{2}xg_{2}^{-1}$ commute for all $g_{1}, g_{2} \in G$. Thus, we may test in $\textsf{AC}^{0}$ whether $\ncl_{G}(x)$ is Abelian. The result follows.
\end{proof}

Next, we show that if $G, H$ are Fitting-free groups with a unique minimal normal subgroup, then we can decide isomorphism between $G$ and $H$ in parallel. Babai, Codenotti, and Qiao established that in this case, $|\Aut(G)| \leq |G|^{O(1)}$ and that $\text{Iso}(G, H)$ can be listed in time $|G|^{O(1)}$ \cite[Cor.~1]{BCQ}=\cite[Theorem~6.4.1]{CodenottiThesis}. We parallelize the latter algorithm as follows.

\begin{proposition}[{cf. \cite[Cor.~1(b)]{BCQ}=\cite[Theorem~6.4.1(b)]{CodenottiThesis}}] \label{thm:6.4.1}
Let $G, H$ be Fitting-free groups of order $n$ with a unique minimal normal subgroup. We can list $\text{ISO}(G, H)$ with an $\textsf{AC}$ circuit of depth $O((\log n)(\log \log n)^2)$ and size $\poly(n)$.
\end{proposition}

\begin{proof}
Grochow and Levet \cite{GrochowLevetWL} showed that we can identify the non-Abelian simple factors of the socle in $\textsf{L}$. Fix a diagonal $\varphi$ for $\Soc(G)$. Let $\mathcal{D}$ be the set of diagonals for $\Soc(H)$. In parallel, consider every possible diagonal $\psi \in \mathcal{D}$. By Corollary~\ref{cor:BCGQ4.1}, it suffices to find:
\[
\bigcup_{\psi \in \mathcal{D}} \text{ISOds}(G, H; \varphi, \psi).
\]
By Lemma~\ref{lem:4.1}, $|\mathcal{D}| \leq |H|^2$. For clarity of the argument, we will consider one such $\psi$. Now given the set of elements in $\Soc(G)$ and $\Fac(\Soc(G))$, we can easily compute $\text{PKer}(G)$ in $\textsf{AC}^{0}$. Furthermore, we can compute $G/\text{PKer}(G)$ in $\textsf{AC}^{0}$. Similarly, we can compute $H/\text{PKer}(H)$ in $\textsf{AC}^{0}$. By Theorem~\ref{thm:5.2.1}, we can list all $f \in \text{PISO}(P, Q)$ using an $\textsf{AC}$ circuit of depth $O((\log n)(\log \log n)^2)$ and size $\poly(n)$. By Lemma~\ref{lem:6.3.10Parallel}, for each such $f$, we can test in $\FOPLL$ if it extends to some $\chi \in \text{ISO}(G, H; \varphi, \psi)$. Thus, we have used an $\textsf{AC}$ circuit of depth $O((\log n)(\log \log n)^2)$ and size $\poly(n)$, as desired.
\end{proof}

\subsection{Embedding into the Automorphism Group of the Socle}
The mathematics in this section is the same as in \cite{BCGQ, BCQ}; the new result we need is Observation~\ref{obs:ComputePermutationRepresentation}, but to even state the result and how it will be used, we must recall much of the previously established mathematics.

Let $G, H$ be Fitting-free groups, and suppose that $\Soc(G) \cong \Soc(H)$. Consider the decompositions of $\Soc(G)$ and $\Soc(H)$ into direct products of minimal normal subgroups:
\begin{align*}
\Soc(G) = \prod_{i=1}^{d} \prod_{j=1}^{z_{i}} N_{i,j} \text{ and } \, \Soc(H) = \prod_{i=1}^{d} \prod_{i=1}^{z_{i}} M_{i,j}.
\end{align*}

Suppose that for all $i, j$, $N_{i,j} \cong M_{i,j} \cong K_{i}$, where $K_{i}$ is characteristically simple. That is, $K_{i} \cong T_{i}^{t_{i}}$ for some non-Abelian simple group $T_{i}$ and some $t_{i} \in \mathbb{N}$. For notational convenience, write $\mathcal{K} = \prod_{i=1}^{d} K_{i}^{z_{i}}$. Let $\gamma, \xi$ be diagonal products of $(N_{i,j}), (M_{i,j})$ respectively. Let $\alpha_{\gamma} \in \text{ISOd}(\Soc(G), \mathcal{K}; \gamma, \Delta)$ and $\beta_{\xi} \in \text{ISOd}(\Soc(H), \mathcal{K}; \xi, \Delta)$, where $\Delta$ is the standard diagonal product.	By Lemma~\ref{CharacterizeSemisimple}, the conjugation action of $G$ and $H$ on their socles yields the respective embeddings $\alpha_{\gamma}^{*} : G \xhookrightarrow{} \Aut(\mathcal{K})$ and $\beta_{\xi}^{*} : H \xhookrightarrow{} \Aut(\mathcal{K})$. Note that in the multipliction table model, it is easy to compute $\alpha_{\gamma}^{*}$ and $\beta_{\xi}^{*}$. 

\begin{observation} \label{obs:ComputePermutationRepresentation}
Let $G$ be a Fitting-free group, given by its multiplication table. Suppose we are also given the direct product decomposition of $\Soc(G)$ into non-Abelian simple factors, where each simple factor is specified by its elements. Fix $\alpha_{\gamma} \in \text{ISOd}(\Soc(G), \mathcal{K}; \gamma, \Delta)$. Then in $\textsf{AC}^{0}$, we can compute $\alpha_{\gamma}^{*}$.
\end{observation}

Let $G^{*} := \alpha_{\gamma}^{*}(G)$ and $H^{*} := \beta_{\xi}^{*}(H)$. In particular, $\Soc(G^{*}) = \Soc(H^{*}) = \Inn(\mathcal{K})$. As the conjugation action of $G$ and $H$ fix their respective minimal normal subgroups, we have in fact that $G^{*}, H^{*} \leq \prod_{i=1}^{d} \Aut(K_{i})^{z_{i}}$. Let:
\[
\text{ISO}(G^{*}, H^{*}) = \{ \widehat{f} : f \in \text{Iso}(\Soc(G), \Soc(H)) \}.
\]

\noindent By Lemma~\ref{CharacterizeSemisimple}, we have that:
\[
\text{Iso}(G, H) = \alpha_{\gamma}^{*} \text{ISO}(G^{*}, H^{*}) (\beta_{\xi}^{*})^{-1}.
\]

\noindent By Corollary~\ref{cor:BCGQ4.1}, we will need to consider diagonal products of $\Soc(G)$ (respectively, $\Soc(H)$) that respect the decomposition of $\Soc(G)$ (respectively, $\Soc(H)$) into a direct product of non-Abelian simple groups. For diagonals $\varphi$ of $\Soc(G)$ and $\psi$ of $\Soc(H)$, let:
\[
\text{ISOds}^{*}(G^*, H^*; \varphi, \psi) = \{ \widehat{f} : f \in \text{ISOd}(\Soc(G), \Soc(H); \varphi, \psi)\}.
\]

\noindent Now let $\varphi$ be a diagonal product of $\Soc(G)$, where we write $\Soc(G)$ as a direct product of non-Abelian simple groups. Define $\psi$ analogously for $\Soc(H)$. Let $\Delta$ be the standard diagonal product of the socles into non-Abelian simple factors. We have the following:
\begin{align} \label{eq:6.4}
\text{ISOds}(G, H; \varphi,\psi) = \alpha_{\gamma}^{*}\text{ISOds}^{*}(G^*, H^*; \Delta, \Delta) (\beta_{\xi}^{*})^{-1}.
\end{align}

\subsection{Reduction to Twisted Code Equivalence} \label{sec:ReductionTwistedCodeEquivalence}
By Corollary~\ref{cor:BCGQ4.1} and Equation~(\ref{eq:6.4}), the problem reduces to finding $\text{ISOds}^{*}(G^{*}, H^{*}, \Delta, \Delta)$, where $\Delta$ is the standard diagonal product of the socles. We will recall from \cite[Section~5.3]{BCQ} and \cite[Lemma~6.5.1]{CodenottiThesis} how to reduce $\text{ISOds}^{*}(G^{*}, H^{*}, \Delta, \Delta)$ to \algprobm{Twisted Code Equivalence}, and show that this reduction can be parallelized.

\begin{lemma}[cf. {\cite[Lemma~6.5.1]{CodenottiThesis}, \cite[Section~5.3]{BCQ}}] \label{lem:6.5.1}
Given Fitting-free groups $G,H$ of order $n$, there are codes $\overline{G}, \overline{H}$ such that computing $\text{ISOds}^*(\overline{G}, \overline{H},\Delta, \Delta)$, where $\Delta$ is the standard diagonal, reduces to \algprobm{Twisted Code Equivalence} by uniform $\cc{AC}$ circuits of depth $O((\log n)(\log \log n)^2)$ and size $\poly(n)$. 

More precisely, with that complexity, $\overline{G}, \overline{H}$ and permutation groups $W_i$ can be computed so that
\[
\text{ISOds}^{*}(\overline{G}, \overline{H}; \Delta, \Delta) = \widehat{\text{EQ}}_{(W_{1}, \ldots, W_{r})}(\overline{G}, \overline{H}),
\]
where $\widehat{\text{EQ}}$ denotes taking the code equivalences $\pi$ (which are permutations of their domains) and lifting them (each lifts uniquely) to isomorphisms of the groups $\overline{G}, \overline{H}$.
\end{lemma}

\begin{proof}
We follow the details from \cite[Lemma~6.5.1]{CodenottiThesis}. Let $G_{ij}^{*}$ be the restriction of $G^{*}$ to the $j$th copy of $K_{i}$, and $H_{ij}^{*}$ the restriction of $H^{*}$ to the $j$th copy of $K_{i}$. For every $i, j, i', j'$, we compute $\text{ISOds}^{*}(G_{ij}^{*}, H_{i'j'}^{*}; \Delta, \Delta)$. As $G_{ij}^{*}, H_{i'j'}^{*}$ each have a unique minimal normal subgroup, we can by Proposition~\ref{thm:6.4.1}, list the isomorphisms between them that preserve the standard diagonals, using an $\textsf{AC}$ circuit of depth $O((\log n)(\log \log n)^2)$ and size $\poly(n)$. Note that there are $O(\log^2 n)$ such pairs $((i,j), (i', j'))$, and so we can make all comparisons in parallel without adding to the circuit depth. 

Let $\Gamma_1, \ldots, \Gamma_r$ be representatives of these isomorphism types, which will be our alphabets. For each pair $(i, j)$, choose an arbitrary isomorphism $\chi_{ij}$ that preserve isomorphism between $G_{ij}^{*}$ and the corresponding representative $\Gamma_{h}$, and let $\delta_{ij}$ be the corresponding isomorphism for $H_{ij}^{*}$. Note that as finite simple groups are $2$-generated, we can select an isomorphism for one of the finite simple direct factors $T$ of $G_{ij}^{*}$ and $\Gamma_{i}$ in $\textsf{L}$ \cite{TangThesis}. Then by considering the conjugation action of $G$, we obtain in $\textsf{AC}^{0}$ an isomorphism $\chi_{ij} : G_{ij}^{*} \cong \Gamma_{i}$. Thus, we can construct $\chi_{ij}$ in $\textsf{L}$. Similarly, in $\textsf{L}$, we can construct the $\delta_{ij}$.

We create codes $\overline{G}, \overline{H}$ over $\Gamma_1, \ldots, \Gamma_r$ as follows. Let $\sigma \in G^{*}$, and let $\sigma_{ij} \in G_{ij}^{*}$ denote the restriction of $\sigma$ to $N_{ij}$. The string associated with $\sigma$ is $\overline{\sigma} = (\chi_{ij}(\sigma_{ij}))$. Now $\overline{G} := \{ \overline{\sigma} : \sigma \in G^{*}\}$. Define $\overline{H}$ analogously by setting $\overline{\pi} := (\delta_{ij}(\pi_{ij}))_{ij}$, for $\pi \in H$. At this stage, we are given the $(\chi_{ij})$ and $(\delta_{ij})$, as well as the $G_{ij}^{*}$ by their multiplication tables. Thus, we can write down each $\overline{\sigma}$ and each $\overline{\pi}$ in $\textsf{AC}^{0}$. 

For each $\ell \in [r]$, let $W_{\ell} = \text{Autds}^{*}(\Gamma_{\ell}; \Delta, \Delta) := \text{ISOds}^{*}(\Gamma_{\ell}, \Gamma_{\ell}; \Delta, \Delta)$ be the group of automorphisms that preserve the standard diagonal. By Lemma~\ref{lem:6.3.10}, each such automorphism is determined by the permutation it induces on the non-Abelian simple factors of the socle. In particular, if $\Gamma_{\ell}$ is acting on a copy of $K_{i} \cong T_{i}^{t_{i}}$ with $T_{i}$ a non-Abelian simple group, then $W_{\ell}$ has a faithful permutation of degree $t_{i}$. By \cite[Theorem~6.4.1(a)]{CodenottiThesis}, each $|W_{i}| \in \poly(n)$. Furthermore, by Proposition~\ref{thm:6.4.1}, we can list each $W_{i}$ with an $\textsf{AC}$ circuit of depth $O((\log n)(\log \log n)^2)$ and size $\poly(n)$.

Let $m_{\ell}$ denote the number of the $G_{ij}^{*}$ that are isomorphic to $\Gamma_{\ell}$. Without loss of generality, suppose that the number of $H_{ij}^{*}$ isomoprhic to $\Gamma_{\ell}$ is also $m_{\ell}$ (otherwise, $G^{*}$ and $H^{*}$ are not isomorphic). Then $\overline{G}, \overline{H} \leq \prod_{\ell=1}^{r} \Gamma_{\ell}^{m_{\ell}}$ are subgroups, and so we view them as group codes as well. 

The correctness of the construction and the equality stated in the lemma comes from \cite[Lemma~6.5.1]{CodenottiThesis} and \cite[Section~5.3]{BCQ}, and we have now established the stated complexity bounds on the reduction.
\end{proof}

\subsection{Algorithm}

We now prove Theorem~\ref{thm:FittingFree}. 

\begin{proof}
By Lemma~\ref{lem:IsFittingFree}, we can check in $\textsf{AC}^{0}$ whether $G, H$ are both Fitting-free. If $G$ is Fitting-free and $H$ is not, then we return that $G \not \cong H$. So assume now that both $G, H$ are Fitting-free. Grochow \& Levet \cite{GrochowLevetWL} showed that we can in $\textsf{L}$, identify the elements of $\Soc(G)$ and decompose $\Soc(G)$ into a direct product of non-Abelian simple factors. In particular, the procedure of Grochow \& Levet yields the non-Abelian simple factors by listing their elements. Given such a decomposition, we can in $\textsf{AC}^{0}$ obtain the minimal normal subgroups of $G$, by considering the conjugation action of $G$ on the non-Abelian simple factors of the $\Soc(G)$. By similar argument, we can decompose $\Soc(H)$ into a direct product of minimal normal subgroups of $H$.

We now turn to fixing a diagonal product $\varphi$ for $\Soc(G)$. We accomplish this by fixing a pair of generators for each non-Abelian simple factor of $\Soc(G)$. As non-Abelian simple groups are $2$-generated, we can identify such a generating sequence for each non-Abelian simple factor of $\Soc(G)$ in $\textsf{L}$ \cite{TangThesis}. Now by Lemma~\ref{lem:4.1}, there are at most $n^2$ diagonal products products of $\Soc(H)$. We may identify each such diagonal product in $\textsf{L}$. Let $\mathcal{D}$ be the set of diagonal products of $\Soc(H)$.

Let $G^{*}$ be the embedding of $G$ into $\Aut(\Soc(G))$, and define $H^{*}$ analogously. By Observation~\ref{obs:ComputePermutationRepresentation}, we can compute $G^{*}, H^{*}$ in $\textsf{AC}^{0}$. By Equation~(\ref{eq:6.4}), the problem now reduces to computing $\text{ISOds}^{*}(G^{*}, H^{*}; \Delta, \Delta)$, where $\Delta$ is the standard diagonal product of the socle. In particular, this reduction yields the $\Gamma_{i}, W_{i}$ ($i \in [r]$) and codes $\overline{G}, \overline{H}$, as defined in Section~\ref{sec:ReductionTwistedCodeEquivalence}. By Lemma~\ref{lem:6.5.1}, this reduction is computable using an $\textsf{AC}$ circuit of depth $O((\log n)(\log \log n)^2)$
and size $\poly(n)$.

Now by Theorem~\ref{thm:TwistedCodeEquivalence} and using the fact that each $|W_{i}| \in \poly(n)$, we can compute $\widehat{EQ}_{(W_1,  \ldots, W_r)}(\overline{G}, \overline{H})$ using an $\textsf{AC}$ circuit of depth $O((\log^2 n) \cdot \poly(\log \log n))$ and size $\poly(n)$. We now recover $\text{ISO}(G, H)$ by computing:
\[
\text{ISO}(G, H) := \bigcup_{\psi \in \mathcal{D}} \alpha_{\varphi}^{*} \text{ISOds}^{*}(G^{*}, H^{*}; \Delta, \Delta)(\beta_{\psi}^{*})^{-1}.
\]
In total, we have an $\textsf{AC}$ circuit of depth $O((\log^{2} n) \cdot \poly(\log \log n))$ and size $\poly(n)$ as desired.
\end{proof}

\section{Isomorphism of Fitting-free permutation groups is at least as hard as Linear Code Equivalence}

In this section, we will establish the following:

\begin{theorem}
There exists a uniform $\textsf{AC}^{0}$ many-one reduction from \algprobm{Linear Code Equivalence}, and hence \algprobm{Graph Isomorphism}, to isomorphism testing of Fitting-free groups, where the groups are given by generating sequences of permutations.
\end{theorem}

\begin{proof}
Petrank \& Roth \cite{PR97} give a reduction from \algprobm{Graph Isomorphism} to \algprobm{Linear Code Equivalence} over $\F_2$. If the starting graphs have $n$ vertices and $m$ edges, the resulting codes are $m$-dimensional subspaces of $\F_2^{3m+n}$. The key to our reduction is thus a reduction from \algprobm{Linear Code Equivalence} over $\F_2$ to (abstract) isomorphism of Fitting-free permutation groups.

Let $\pi \colon S_5^N \to (\mathbb{Z}/2\mathbb{Z})^N$ be the coordinate-wise sign homomorphism $(\sigma_1,\dotsc,\sigma_N) \stackrel{\pi}{\mapsto} (\text{sgn}(\sigma_1),\dotsc,\text{sgn}(\sigma_N))$. Given a code $C \leq \mathbb{F}_2^N$, define a group $G_C := \pi^{-1}(C)$, the preimage of $C \leq \F_2^N$ under $\pi$. We first show how to construct a small generating set of permutations for $G_C$. Since $G_C \leq S_5^N$, we will build our permutations as elements of $S_5^N$, thought of as a subgroup of $S_{5N}$ in the natural way. For each of the $N$ normal copies of $A_5$, we include a size-2 generating set of $A_5$ as part of our generating set of $G_C$. Let $v_1, \dotsc, v_k \in \F_2^N$ be a basis for $C$ over $\F_2$; for the purposes of the reduction, we are given $C$ by a generator matrix, and we can take $v_1, \dotsc, v_k$ to be the rows of that matrix. For each such vector $v_i$, let $\sigma_i =((\sigma_i)_1, \dotsc, (\sigma_i)_N) \in S_5^N$ be the permutation in $S_5^N$ whose $j$-th coordinate is
\[
(\sigma_i)_j = 
\begin{cases}
(12) & (v_i)_j = 1 \\
\text{id} & (v_i)_j = 0.
\end{cases}
\]
It is readily verified that the generating set just described consists of $2N+k$ elements, and that they generate $G_C$. Furthermore, $G_C$ is independent of the choice of basis (or even spanning set) for $C$: if another spanning set is used in building the generating set, it will still generate the same (more than just isomorphic) subgroup of $S_5^N$.

Finally, we show that if $C' \leq \F_2^N$ is another code, then $G_C \cong G_{C'}$ iff $C$ and $C'$ are equivalent linear codes. Suppose $G_C \cong G_{C'}$. Then by \cite[Lemma~2.1]{BCGQ}, the isomorphism is induced by an element of $\Aut(A_5) \wr S_N = S_5^N \rtimes S_N$, acting on $\Soc(G_C) = A_5^N$. We note that conjugating by $S_5^N$ does not change $\pi(G_C) = C$. The permutation action of $S_N$ on the factors of $\Soc(G_C)$ induces the same permutation on the factors of $S_5^N$, which corresponds exactly to permuting the columns of the code $C$. Thus, if $\alpha \in S_N$ gives an isomorphism $G_C \to G_{C'}$ as just described, then $C \alpha$ is the same code as $C'$.

Conversely, if $AC \alpha = C'$, reading the preceding paragraph backwards we get that $\alpha$ gives an isomorphism $G_C \to G_{C'}$.
\end{proof}

The previous reduction lets us show that a $c^n$ algorithm for \algprobm{Linear Code Equivalence} over $\F_2$ in fact follows from the polynomial-time algorithm for isomorphism of Fitting-free groups $G$ when $G = \text{PKer}(G)$ \cite{BCGQ}. Rather than Babai's $c^n$ algorithm for  \algprobm{Linear Code Equivalence} \cite{BCGQ} being seen as separate, this shows a $c^n$ algorithm (albeit for a worse $c$---240 versus Babai's 2) over $\F_2$ in fact already followed from the algorithm for \algprobm{Group Isomorphism} for the corresponding groups.

\begin{corollary}
If isomorphism of Fitting-free groups of order $N$ can be solved in time $t(N)$, then \algprobm{Linear Code Equivalence} for codes in $\F_2^n$ can be solved in time $t(240^n)$.
\end{corollary}

\begin{proof}
Given a code $C \leq \F_2^n$ of dimension $k$, the group constructed in the proof above has order $N = 120^n \times 2^k \leq 120^n \times 2^n$.
\end{proof}

\section{Descriptive complexity of Fitting-free groups}

In this section, we will establish the following.

\begin{theorem} \label{thm:DescriptiveComplexity}
If $G$ is a Fitting-free group of order $n$, then $G$ is identified by some $\textsf{FO}$ formula using $O(\log \log n)$ variables.
\end{theorem}

\noindent We contrast this with the previous work of Grochow \& Levet \cite{GrochowLevetWL}, who established that there exists an infinite family of Abelian groups that are not identified by $\textsf{FO}$ formulas with $o(\log n)$ variables. In light of the equivalence between $\textsf{FO}$ and count-free Weisfeiler--Leman \cite{ImmermanLander1990, CFI, GrochowLevetWL}, it suffices to establish the following.

\begin{theorem}[cf. {\cite[Lemma~5.4.4]{BrachterThesis}}] \label{thm:FittingFreeWL}
Let $G$ be a Fitting-free group of order $n$. Then the count-free $O(\log \log n)$-WL Version I will distinguish $G$ from any non-isomorphic group $H$, using $O(\log \log n)$ rounds.
\end{theorem}

In his dissertation, Brachter \cite[Lemma~5.4.4]{BrachterThesis} established the analogous result for the counting variant of the more powerful WL Version II, without controlling for rounds. We will essentially show that Brachter's pebbling strategy from the counting game can be implemented with the corresponding count-free game. In order to obtain $O(\log \log n)$ rounds with count-free WL Version I, we will also take advantage of the following result due to Babai, Kantor, and Lubotsky. 

\begin{theorem}[{\cite{BabaiKantorLubotsky}}] \label{thm:BabaiKantorLubotsky}
There exists an absolute constant $C$, such that for every finite simple group $S$, there exists a generating set $g_1, \ldots, g_7 \in S$ such that every element of $S$ can be realized as a word over $\{g_1, \ldots, g_7\}$ of length at most $C \log |S|$.
\end{theorem}

The Classification of Finite Simple Groups provides that every finite simple group is generated by at most two elements. However, there are no bounds in general of the length of the words for such generating sets. 

In light of Theorem~\ref{thm:FittingFreeWL} and the parallel implementation of count-free WL due to Grohe and Verbitsky \cite{GroheVerbitsky}, we also obtain the following result.

\begin{corollary} \label{cor:quasiFOLL}
Let $G$ be a Fitting-free group of order $n$, and let $H$ be arbitrary. We can decide isomorphism between $G$ and $H$ in $\exists^{\log n \log \log n}\textsf{FOLL}$, which can be simulated by $\textsf{quasiFOLL}$ circuits of size $n^{O(\log \log n)}$ (by trying all $n^{O(\log \log n)}$ possibilities in parallel).
\end{corollary}

\begin{proof}[Proof (Sketch).]
The proof of Theorem~\ref{thm:FittingFreeWL} shows that there exist $O(\log \log n)$ group elements that, after individualizing, $O(1)$-dimensional count-free WL Version I will assign a unique color to each group element after $O(\log \log n)$ rounds. We use $O(\log n \log \log n)$ existentially-quantified non-deterministic bits to guess $O(\log \log n)$ such group elements to individualize (this is the $\exists^{\log n \log \log n}$). We then run $O(1)$-dimensional count-free WL Version I for $O(\log \log n)$ rounds, which is $\textsf{FOLL}$-computable \cite{GroheVerbitsky}.
\end{proof}

Collins and Levet \cite{CollinsLevetWL} previously asked if the isomorphism problem for Fitting-free groups was decidable using a $\textsf{quasiFOLL}$ circuit of size $n^{O(\log \log n)}$, in tandem with $O(\log n)$ universally quantified co-nondeterministic bits and a single $\textsf{Majority}$ gate. Our result answers this question, and in fact shows that we do not need the co-nondeterministic bits or the $\textsf{Majority}$ gate.

\subsection{Preliminaries}

We will now recall some preliminaries from \cite{GLDescriptiveComplexity, CollinsLevetWL}.

\begin{definition}[cf. {\cite[Definition~4.9]{GLDescriptiveComplexity}}]
Let $G$ be a Fitting-free group. Let $\Soc(G) = S_{1} \times \cdots \times S_{k}$, where each $S_{i}$ is a non-Abelian simple normal subgroup of $\Soc(G)$. For any $s \in \Soc(G)$, write $s = s_{1} \cdots s_{k}$, where each $s_{i} \in S_{i}$. Define the \textit{weight} of $s$, denoted $\wt(s)$, as the number of $i$'s such that $s_{i} \neq 1$. For $s \not \in \Soc(G)$, we define $\wt(s) = \infty$.
\end{definition}

As the $S_{i}$ are the unique subsets of $\Soc(G)$ that are simple, normal subgroups of $\Soc(G)$, the decomposition of $s = s_{1} \cdots s_{k}$ is unique up to the order of the factors. Thus, the definition of weight is well-defined.

Grochow and Levet \cite{GrochowLevetWL} generalized weight to the notion of rank.

\begin{definition}[{\cite[Definition~4.1]{GrochowLevetWL}}]
Let $C \subseteq G$ be a subset of a group $G$ that is closed under taking inverses. We define the $C$-\textit{rank} of $g \in G$, denoted $\rk_{C}(g)$, as the minimum $m$ such that $g$ can be written as a word of length $m$ in the elements of $C$. If $g$ cannot be so written, we define $\rk_{C}(g) = \infty$.
\end{definition}

\begin{lemma}[{\cite[Rank Lemma~4.3]{GrochowLevetWL}}] \label{lem:RankLemma}
Let $C \subseteq G$ be a subset of $G$ that is closed under taking inverses. Suppose that if $x \in C$ and $y \not \in C$, then Spoiler can win in the count-free Version I pebble game from the initial configuration $x \mapsto y$, using $k$ pebbles and $r$ rounds. If $\rk_{C}(g) \neq \rk_{C}(h)$, then Spoiler can win with $k+1$ pebbles and $r + \log(d) + O(1)$ rounds, where $d = \text{diam}(\text{Cay}(\langle C \rangle, C)) \leq |\langle C \rangle| \leq |G|$.
\end{lemma}

Collins and Levet showed that the count-free WL Version I can quickly distinguish the weight one elements. We state this result in the pebble game characterization.

\begin{lemma}[see {\cite[Lemma~6.10]{CollinsLevetWL}}] \label{lem:CL6.10}
Let $G, H$ be Fitting-free groups. Let $S$ be a non-Abelian simple direct factor of $\Soc(G)$. Suppose that $g \in S$ and $\wt(h) > 1$. Then Spoiler can win in the count-free Version I pebble game, starting from the initial configuration $g \mapsto h$, using $O(1)$ pebbles and $O(\log \log n)$ rounds.\footnote{The statement of \cite[Lemma~6.10]{CollinsLevetWL} assumes only $g \in S$ and $h \not \in \Soc(H)$. However, their proof shows that if $g \in S$ and $\wt(h) > 1$, then Spoiler can win with $O(1)$ pebbles and $O(\log \log n)$ rounds.}
\end{lemma}
  
We now obtain the following immediate consequence.

\begin{lemma}
Let $G, H$ be Fitting-free groups. Suppose that $g \in G, h \in H$ with $\wt(g) \neq \wt(h)$. Then in the Version I pebble game, starting from the initial configuration $g \mapsto h$, Spoiler can win using $O(1)$ pebbles and $O(\log \log n)$ rounds.
\end{lemma}

\begin{proof}
Let $C_{G}$ be the set of weight one elements in $G$, and define $C_{H}$ analogously. By Lemma~\ref{lem:CL6.10}, if $g \in C_{G}$ and $h \not \in C_{H}$, then Spoiler can win with $O(1)$ pebbles and $O(\log \log n)$ rounds. The hypotheses of Lemma~\ref{lem:RankLemma} are now satisfied. Thus, if $\wt(g) \neq \wt(h)$, then by Lemma~\ref{lem:RankLemma}, Spoiler can win with $O(1)$ pebbles and $O(\log \log n)$ rounds. The result now follows.
\end{proof}

\subsection{Proof of Theorem~\ref{thm:FittingFreeWL}}

We now turn to proving Theorem~\ref{thm:FittingFreeWL}. 

\begin{definition}[Pebbling Strategy] \label{def:InitialPebbling}
Let $G = \prod_{i=1}^{k} S_{i}$ be a group of order $n$, where each $S_{i}$ is non-Abelian simple. For each $i \in [k]$, let $(x_{i,1}, \ldots, x_{i,7})$ be the generating sequence prescribed by Theorem~\ref{thm:BabaiKantorLubotsky}. Now for each $j \in [7]$, Spoiler pebbles $x_{j} := \prod_{i=1}^{k} x_{i,j}$. Let $\ell := \lceil \log_{2}(k) \rceil$. Suppose that Spoiler places the following pebbles:
\begin{align*}
&p_{1} := \prod_{i=1}^{k/2} x_{i}, \\
&p_{2} := \left( \prod_{i=1}^{k/4} x_{i} \right) \cdot \left( \prod_{i=k/2+1}^{3k/4} x_{i} \right), \\
&p_{3} = \left( \prod_{a=1}^{\frac{k}{8}}  x_a\right)\left( \prod_{a=\frac{k}{4} + 1}^{\frac{3k}{8}}  x_a\right)\cdots \left( \prod_{a=\frac{3k}{4} + 1}^{\frac{7k}{8}}  x_a\right),\\
\end{align*}

\noindent and so on. More generally, for each $j \in [\ell]$, Spoiler pebbles:
\begin{align*}
p_{j} := \prod_{i=1}^{j} \left(\prod_{m=1+(i-1) \cdot k/2^{i-1}}^{k/2^{i}+(i-1)k/2^{j}} x_{m}\right).
\end{align*}
\end{definition}

\begin{lemma} \label{lem:Ordering} 
Let $G = \prod_{i=1}^{k} S_{i}$ and $H = \prod_{i=1}^{k} S_{i}'$ be groups of order $n$, where each $S_{i}, S_{i}'$ is non-Abelian simple (note that $S_{i}$ and $S_{i}'$ need not be isomorphic). Let $\ell := \lceil \log_{2}(k) \rceil$. Suppose that Spoiler begins by employing the pebbling strategy from Definition~\ref{def:InitialPebbling}. We denote Duplicator's responses by:
\begin{itemize}
    \item for each $j \in [7]$, $x_{j}' \in H$ in response to $x_{j} \in G$, and
    \item for each $q \in [\ell]$, $p_{q}'$ in response to $p_{q}$.    
\end{itemize}

\noindent Suppose that $\wt(x_{j}) = \wt(x_{j}')$ for all $j \in [7]$; and that for all $q \in [\ell]$, $\wt(p_{q}) = \wt(p_{q}')$. Then for any $a,b \in [k]$ with $a \neq b$, there exists $m \in [\ell]$ such that $\wt(p_{m} \cdot x_{a}^{-1}) \neq \wt(p_{m} \cdot x_{b}^{-1})$.
\end{lemma}

The idea is intuitively that we can use the $p_i$ to do binary search to find out which factor $x_a$ and $x_b$ live in.

\begin{proof}
The proof is by induction on $k$, the number of direct factors of $G$ and $H$. When $k = 1$, $G$ and $H$ are both non-Abelian simple. So the claim is vacuously true. Fix $k \geq 1$, and suppose that for all $1 \leq h \leq k$, if $1 \leq a < b \leq h$, then there exists some $m \in [\ell]$ such that $\wt(p_{m} \cdot x_{a}^{-1}) \neq \wt(p_{m} \cdot x_{b}^{-1})$.

Consider now the $k+1$ case. If $\wt(p_{1} \cdot x_{a}^{-1}) \neq \wt(p_{1} \cdot x_{b}^{-1})$, then we are done. Otherwise, we have that either $S_{a}, S_{b} \in \{S_{1}, \ldots, S_{k/2}\}$ or $S_{a}, S_{b} \in \{ S_{k/2 +1}, \ldots, S_{k}\}$. Without loss of generality, suppose that $S_{a}, S_{b} \in \{S_{1}, \ldots, S_{k/2}\}$. Let $G' := \prod_{i=1}^{k/2} S_{i}$ and $H' := \prod_{i=1}^{k/2} S_{i}'$. We thus apply the inductive hypothesis to $G', H'$, to deduce that there exists $2 \leq m \leq \ell$ such that $\wt(p_{m} \cdot x_{a}^{-1}) \neq \wt(p_{m} \cdot x_{b}^{-1})$. The result now follows. 
\end{proof}

\begin{lemma} \label{lem:PreservesOrder}
Let $G = \prod_{i=1}^{k} S_{i}$ and $H = \prod_{i=1}^{k} S_{i}'$ be groups of order $n$, where each $S_{i}, S_{i}'$ is non-Abelian simple (note that $S_{i}$ and $S_{i}'$ need not be isomorphic). Let $\ell := \lceil \log_{2}(k) \rceil$. Suppose that Spoiler begins by employing the pebbling strategy from Definition~\ref{def:InitialPebbling}. We denote Duplicator's responses by:
\begin{itemize}
    \item For each $j \in [7]$, Duplicator pebbles $x_{j}' \in H$ in response $x_{j} \in G$, and 
    \item For each $q \in [\ell]$, Duplicator pebbles $p_{q}'$ in response to $p_{q}$.    
\end{itemize}

\noindent Suppose that $\wt(x_{j}) = \wt(x_{j}')$ for all $j \in [7]$; and that for all $q \in [\ell]$, $\wt(p_{q}) = \wt(p_{q}')$. For each $i \in [k]$ and each $j \in [7]$, let $s_{i,j}'$ be the unique element in $S_{i}'$ s.t. $\wt(x_{j}' \cdot (s_{i,j}')^{-1}) < \wt(x_{j}')$. Suppose that any of the following conditions hold:
\begin{enumerate}[label=(\alph*)]

\item For some $a \neq b$, Spoiler pebbles $g \in S_{a}$ and Duplicator responds by pebbling $h \in S_{b}'$ (or vice-versa),

\item Suppose that the map $(s_{i,1}, \ldots, s_{i,7}) \mapsto (s_{i,1}', \ldots, s_{i,7}')$ does not extend to an isomorphism of $S_{i} \cong S_{i}'$.

\item For some $i \in [k]$, $S_{i} \not \cong S_{i}'$.

\item Consider the same assumptions as in (b). Let $\varphi_{i} : S_{i} \cong S_{i}'$ be the isomorphism induced by the map $(s_{i,1}, \ldots, s_{i,7}) \mapsto (s_{i,1}', \ldots, s_{i,7}')$. Suppose that Spoiler pebbles some $g \in S_{i}$ and Duplicator responds by pebbling $h \in S_{i}'$. Now suppose that $\varphi_{i}(g) \neq h$. 
\end{enumerate} 

Then Spoiler can win with $O(1)$ additional pebbles and $O(\log \log n)$ additional rounds.
\end{lemma}

\begin{proof}
\noindent 
\begin{enumerate}[label=(\alph*)]
\item Spoiler begins by pebbling $s_{a,1}, \ldots, s_{a,7}$. Let $h_{1}', \ldots, h_{7}'$ be Duplicator's response. We have several cases to consider. 
\begin{itemize}
    \item \textbf{Case 1:} If $h_{1}' \neq s_{a,1}'$, then by Lemma~\ref{lem:Ordering}, there exists $m \in [\ell]$ such that $\wt(p_{m} \cdot s_{a,1}^{-1}) \neq \wt(p_{m}' \cdot (h_{1}')^{-1})$. So by Lemma~\ref{lem:RankLemma}, Spoiler can win with $O(1)$ additional pebbles and $O(\log \log n)$ additional rounds.

    \item \textbf{Case 2:} Suppose instead that  the map $(g, s_{a,1}, \ldots, s_{a,7}) \mapsto (h, h_{1}', \ldots, h_{7}')$ does not extend to an isomorphism. Collins and Levet \cite[Section~3]{CollinsLevetWL} previously established that, Spoiler can win with $O(1)$ additional pebbles and $O(\log \log n)$ additional rounds.

    \item \textbf{Case 3:} Finally, suppose that neither Cases 1-2 hold. So the map $(g, s_{a,1}, \ldots, s_{a,7}) \mapsto (h, h_{1}', \ldots, h_{7}')$ does extend to an isomorphism. However, $h \in S_{b}'$, while $h_{1}' = s_{a,1}' \in S_{a}'$. So $\wt(hh_{1}') = 2$. On the other hand, $g, s_{a,1} \in S_{a}$, and so $\wt(gs_{a,1}) = 1$. By Lemma~\ref{lem:RankLemma}, Spoiler now wins with $O(1)$ additional pebbles and $O(\log \log n)$ additional rounds. 
    
\end{itemize}

\item Spoiler begins by pebbling $s_{i,1}, \ldots, s_{i,7}$. By part (a) and Lemma~\ref{lem:Ordering}, we may assume that Duplicator responds by pebbling $s_{i,1}', \ldots, s_{i,7}'$. Otherwise, Spoiler wins with $O(1)$ additional pebbles and $O(\log \log n)$ additional rounds. Now suppose that the map $(s_{i,1}, \ldots, s_{i,7}) \mapsto (s_{i,1}', \ldots, s_{i,7}')$ does not extend to an isomorphism of $S_{i} \cong S_{i}'$. Collins and Levet \cite[Section~3]{CollinsLevetWL} previously established that, Spoiler can win with $O(1)$ additional pebbles and $O(\log \log n)$ additional rounds.

\item This follows immediately from (b).

\item This follows immediately from (b). \qedhere
\end{enumerate}
\end{proof}

\noindent We will now prove Theorem~\ref{thm:FittingFreeWL}.

\begin{proof}[Proof of Theorem~\ref{thm:FittingFreeWL}]
Collins and Levet \cite[Lemma~6.6]{CollinsLevetWL} established that if $G$ is Fitting-free and $H$ is not Fitting-free, then Spoiler can win with $O(1)$ pebbles and $O(1)$ rounds. So suppose that $H$ is Fitting-free.    
 
Let $\Fac(\Soc(G)) = \{ S_{1}, \ldots, S_{k} \}$ and $\Fac(\Soc(H)) = \{ S_{1}', \ldots, S_{k'}'\}$. We now claim that if $k \neq k'$, then Spoiler can win with $O(1)$ additional pebbles and $O(1)$ additional rounds. Suppose that, without loss of generality, $k > k'$. In this case, Spoiler pebbles some element $g \in G$ such that $\wt(g) = k$. Regardless of which $h \in H$ Duplicator pebbles in response, $\wt(h) \neq \wt(g)$. So by Lemma~\ref{lem:RankLemma}, Spoiler wins with $O(1)$ additional pebbles and $O(\log \log n)$ additional rounds.

Let $\ell := \lceil \log_{2}(k) \rceil$. Spoiler now applies the pebbling strategy of Definition~\ref{def:InitialPebbling} to $\Soc(G)$, to pebble $x_{1}, \ldots, x_{7}, p_{1}, \ldots, p_{\ell}$.  Let $x_{1}', \ldots, x_{7}', p_{1}', \ldots, p_{\ell}'$ be Duplicator's response. By Lemma~\ref{lem:RankLemma}, we may assume that:
\begin{itemize}
    \item for all $j \in [7]$, $\wt(x_{j}) = \wt(x_{j}')$, and 
    \item for all $h \in [\ell]$, $\wt(p_{h}) = \wt(p_{h}')$.
\end{itemize}

\noindent Otherwise, Spoiler can win with $O(1)$ additional pebbles and $O(\log \log n)$ additional rounds. By Lemma~\ref{lem:PreservesOrder}, we may assume that, for all $i \in [k]$, $S_{i} \cong S_{i}'$. 

For each $i \in [k]$ and each $j \in [7]$, let $s_{i,j}'$ be the unique element in $S_{i}'$ s.t. $\wt(x_{j}' \cdot (s_{i,j}')^{-1}) < \wt(x_{j}')$. By Lemma~\ref{lem:PreservesOrder}, we have that for each $i \in [k]$, the map $(s_{i,1}, \ldots, s_{i,7}) \mapsto (s_{i,1}', \ldots, s_{i,7}')$ extends to an isomorphism of $S_{i} \cong S_{i}'$. Otherwise, Spoiler wins with $O(1)$ additional pebbles and $O(1)$ additional rounds.

For each $i \in [k]$, let $\varphi_{i} : S_{i} \to S_{i}'$ by the isomorphism induced by the map $(s_{i,1}, \ldots, s_{i,7}) \mapsto (s_{i,1}', \ldots, s_{i,7}')$. Let $\varphi : \Soc(G) \to \Soc(H)$ be the unique isomorphism satisfying $\varphi|_{S_{i}} = \varphi_{i}$, for all $i \in [k]$. By Lemma~\ref{CharacterizeSemisimple}, $\varphi$ extends in at most one way to an isomorphism $\widehat{\varphi}$ of $G$ and $H$. In particular, as $G \not \cong H$, there exists $g \in G$, such that for all $h \in H$, there exists $s_{gh} \in \Soc(G)$ such that $\varphi(gs_{gh}g^{-1}) \neq h\varphi(s_{gh})h^{-1}$. Grochow and Levet \cite[Theorem~4.17]{GLDescriptiveComplexity} showed that there exists such an $s_{gh}$ satisfying $\wt(s_{gh}) = 1$.

Spoiler begins by pebbling $g \in G$. Let $h$ be Duplicator's response. Let $s_{gh} \in \Soc(G)$ such that $\varphi(gs_{gh}g^{-1}) \neq h\varphi(s_{gh})h^{-1}$. Without loss of generality, suppose that $s_{gh} \in S_{1}$. Spoiler pebbles $s_{gh}$. Let $s'$ be Duplicator's response. By Lemma~\ref{lem:CL6.10}, we may assume that $\wt(s_{gh}) = \wt(s') = 1$; otherwise, Spoiler wins with $O(1)$ additional pebbles and $O(\log \log n)$ additional rounds. By Lemma~\ref{lem:PreservesOrder}, we have that $s' \in S_{1}'$ and $\varphi_{i}(s_{gh}) = s'$; otherwise, Spoiler wins with $O(1)$ additional pebbles and $O(\log \log n)$ additional rounds. 

Finally, suppose that $\varphi(gs_{gh}g^{-1}) \neq hs'h^{-1}$. Spoiler pebbles $s_{1,1}, \ldots, s_{1,7}$; and by Lemma~\ref{lem:PreservesOrder}, wins with $O(1)$ additional pebbles and $O(\log \log n)$ additional rounds.
\end{proof}

\begin{remark}
After developing the above proof, we learned that in fact the socle is $O(\log \log n)$-generated, as any direct product of $t$ non-Abelian simple groups is $O(\log t)$-generated \cite{wiegold, THEVENAZ1997352}. Thus an alternative proof is to pebble a generating set of the socle of size $O(\log \log n)$, and then directly apply the last three paragraphs of the above proof. However, we have kept the above proofs in this section intact as the results throughout the rest of the section may have further applications in other groups where the direct product decomposition is unique as subsets (not just up to isomorphism). Additionally, our proof yields bounds of $\exists^{\log n \log \log n}\textsf{FOLL}$ (Corollary~\ref{cor:quasiFOLL}). 

It is unclear whether count-free WL could be used to achieve these bounds, simply by pebbling $O(\log \log n)$ generators. We first consider WL Version II, which differs from WL Version I only in the initial coloring, where $\bar{u} = (u_1, \ldots, u_{k}), \bar{v} = (v_1, \ldots, v_k)$ receive the same initial color in WL Version II if and only if the map $\bar{u} \mapsto \bar{v}$ extends to an isomorphism of $\langle u_1, \ldots, u_k \rangle$ and $\langle v_1, \ldots, v_k \rangle$ \cite{WLGroups}. The refinement step for WL Versions I and II is identical. Note that the initial coloring of WL Version II is $\textsf{L}$-computable \cite{TangThesis}. Thus, we can achieve bounds of $\exists^{\log n \log \log n} \textsf{L}$ for isomorphism testing of Fitting-free groups using WL Version II. Grochow and Levet \cite[Theorem~6.1]{GrochowLevetWL} established the analogous result using $\textsf{quasiSAC}^{1}$ circuits of size $n^{\Theta(\log \log n)}$ by analyzing the individualization and refinement paradigm on WL Version II. However, a careful analysis of their proof already yields the $\exists^{\log n \log \log n}\textsf{L}$ bound. Note that $\textsf{L}$ and hence, $\exists^{\log n \log \log n}\textsf{L}$, can compute $\algprobm{Parity}$, while $\exists^{\log n \log \log n}\textsf{FOLL}$ cannot compute \algprobm{Parity}. It is open whether $\textsf{FOLL} \subseteq \textsf{L}$, and hence open whether $\exists^{\log n \log \log n}\textsf{FOLL} \subseteq \exists^{\log n \log \log n}\textsf{L}$.

Alternatively, we can use $O(\log n)$ rounds of count-free WL Version I to decide if the map $\bar{u} \mapsto \bar{v}$ extends to an isomorphism of $\langle u_1, \ldots, u_k \rangle$ and $\langle v_1, \ldots, v_k \rangle$. We will briefly sketch the argument. Let $f : \langle u_1, \ldots, u_k \rangle \to \langle v_1, \ldots, v_k \rangle$ be the induced map on the generated subgroups. Suppose that $\langle u_1, \ldots, u_k \rangle \not \cong \langle v_1, \ldots, v_k \rangle$. So there exists a minimum word $\omega(w_1, \ldots, w_k) = w_{i_{1}} \cdots w_{i_{\ell}}$ where $f(\omega(u_1, \ldots, u_k)) \neq \omega(v_1, \ldots, v_k)$. Spoiler pebbles $g := \omega(u_1, \ldots, u_k)$, and let $h$ be Duplicator's response. By the minimality of $\omega(w_1, \ldots, w_k)$, $h \neq v_{i_{1}} \cdots v_{i_{\ell}}$. At the next two rounds, Spoiler pebbles $g' := u_{i_{1}} \cdots u_{i_{\ell/2}}$ and $g'' := u_{i_{\ell/2 +1}} \cdots u_{i_{\ell}}$. Let $h', h''$ be Duplicator's response. Necessarily, either $h' \neq v_{i_{1}} \cdots v_{i_{\ell/2}}$ or $h'' \neq v_{i_{\ell/2+1}} \cdots v_{i_{\ell}}$. Without loss of generality, suppose that $h' \neq v_{i_{1}} \cdots v_{i_{\ell/2}}$. Spoiler now iterates, starting from $g' \mapsto h'$, reusing pebbles on $g, g''$. Spoiler wins with $O(1)$ additional pebbles and $O(\log n)$ additional rounds.

Note that for fixed $k$, each round of count-free WL Version I is $\textsf{AC}^{0}$-computable. Thus, count-free WL Version I can detect in $\textsf{AC}^{1}$ whether for prescribed generators $\bar{g}$ of $\Soc(G)$ and $\bar{h}$ of $\Soc(H)$, the map $\bar{g}\mapsto \bar{h}$ extends to an isomorphism of $\Soc(G) \cong \Soc(H)$. If so, then only $O(1)$ additional rounds are required to decide whether $G \cong H$. In total, this yields bounds of $\exists^{\log n \log \log n}\textsf{AC}^{1}$, to guess generators of $\Soc(G)$ (resp. $\Soc(H)$) to individualize, and then an $\textsf{AC}^{1}$ circuit to run count-free WL Version I. This is worse than $\exists^{\log n \log \log n}\textsf{FOLL}$. 
\end{remark}

\section{Conclusion}

We established several results concerning the complexity of identifying Fitting-free groups. In the multiplication table model, we showed that isomorphism testing between a Fitting-free group $G$ and an arbitrary group $H$ is solvable in $\textsf{AC}^{3}$. When the groups are given by generating sequences of permutations, we established a reduction from \algprobm{Linear Code Equivalence} (and hence, \algprobm{Graph Isomorphism}) to isomorphism testing of Fitting-free groups where $G = \text{PKer}(G)$. 

Finally, we showed that Fitting-free groups are identified by $\textsf{FO}$ formulas with $O(\log \log |G|)$ variables. In contrast, Grochow and Levet \cite{GrochowLevetWL} previously established that there exists an infinite family of Abelian groups that are not identified by $\textsf{FO}$ formulas using $o(\log n)$ variables. 

Our work leaves open several problems.

\begin{question}
Let $G$ be a Fitting-free group, and let $H$ be an arbitrary group. Suppose $G, H$ are given by their multiplication tables. Can isomorphism between $G$ and $H$ be decided in $\textsf{L}$? 
\end{question}

Isomorphism testing of many natural families of graphs, including graphs of bounded treewidth \cite{ElberfeldSchweitzer} and graphs of bounded genus \cite{ElberfeldKawarabayashiGenus}, are $\textsf{L}$-complete. While \algprobm{GpI} is not $\textsf{L}$-complete under $\textsf{AC}^{0}$-reductions \cite{ChattopadhyayToranWagner, CGLWISSAC}, establishing an upper bound of $\textsf{L}$ for isomorphism testing of Fitting-free groups nonetheless remains a natural target. 

While \algprobm{Graph Isomorphism} is solvable in quasipolynomial time \cite{BabaiGraphIso}, the best bound for \algprobm{Linear Code Equivalence} is $(2+o(1))^{n}$ \cite{BCGQ} (\cite{BBBDLLW} give randomized and quantum algorithms that are slightly better, but still $(c + o(1))^n$ time for some $2 > c > 1$). In light of Theorem~\ref{thm:GpIPermutationGroups}, \algprobm{Linear Code Equivalence} remains a barrier to improving the complexity for \algprobm{GpI} in the setting of permutation groups, and solving it in time $2^{o(n)}$ remains an important challenge.

\begin{question}
Establish an analogue of Theorem~\ref{thm:DescriptiveComplexity} for a logic whose distinguishing power is strictly greater than $\textsf{FO}$. Furthermore, in the same logic, exhibit an infinite family of solvable groups $(G_{m})_{m \geq 1}$ and a constant $C$, where any formula identifying $G_{m}$ requires at least $|C| \log |G_{m}|$ variables.
\end{question}

Recall that $\textsf{FO}$ corresponds to count-free WL, and $\textsf{FO} + \textsf{C}$ (first-order logic with counting quantifiers) corresponds to the classical counting WL (see e.g., \cite{CFI, ImmermanLander1990, WLGroups}). Lower bounds against the classical counting WL remain an intriguing open question in the setting of groups. It would be of interest to investigate in the setting of groups, measures of descriptive complexity whose distinguishing power sit between $\textsf{FO}$ and $\textsf{FO} + \textsf{C}$. 

\section*{Acknowledgements}

ML wishes to thank Peter Brooksbank, Edinah Gnang, Takunari Miyazaki, and James B. Wilson for helpful discussions. Parts of this work began at Tensors Algebra-Geometry-Applications (TAGA) 2024. ML wishes to thank Elina Robeva, Christopher Voll, and James B. Wilson for organizing this conference. ML was partially supported by CRC 358 Integral Structures in Geometry and Number Theory at Bielefeld and Paderborn, Germany; the Department of Mathematics at Colorado State University; James B. Wilson's NSF grant DMS-2319370; travel support from the Department of Computer Science at the College of Charleston; and SURF Grant SU2024-06 from the College of Charleston. 

DJ was partially supported by the Department of Computer Science at the College of Charleston, SURF Grant SU2024-06 from the College of Charleston, and ML startup funds.

JAG was partially supported by NSF CAREER award CCF-2047756.

\bibliographystyle{alphaurl}
\bibliography{references}

\end{document}